\documentclass[12pt,onecolumn,letter]{IEEEtran}

\usepackage[dvipdfmx]{graphicx}
\usepackage{amssymb}
\usepackage{amsmath}
\usepackage{latexsym}
\usepackage{colortbl}
\usepackage{multirow}
\usepackage{bm}
\usepackage{multirow}

\def\bW{\bm{W}}
\def\bw{\bm{w}}
\def\bV{\bm{V}}
\def\btheta{\bm{\theta}}
\def\bTheta{\bm{\Theta}}
\def\bxi{\bm{\xi}}

\def\RR{\mathbb R}

\def\NN{\mathbb N}

\def\argmax{\mathop{\rm argmax}}

\def\QED{\mbox{\rule[0pt]{1.5ex}{1.5ex}}}

\def\endproof{\hspace*{\fill}~\QED\par\endtrivlist\unskip}

\newcommand{\qed}{\hfill \QED}

 \newenvironment{proofof}[1]{\vspace*{5mm} \par \noindent
         \quad{\it Proof of #1:\hspace{2mm}}}{\qed
}

\newtheorem{lemma}{Lemma}
\newtheorem{theorem}{Theorem}

\newtheorem{rem}{Remark}
\newtheorem{proposition}{Proposition}
\newtheorem{example}{Example}

\def\Label#1{\label{#1}\ \text{[\ #1\ ]}\ }
\def\Label{\label}

\begin{document}
\title{Universal channel coding\\ for general output alphabet}
\author{Masahito Hayashi~\IEEEmembership{Fellow,~IEEE}
\thanks{This research was partially supported by 
the MEXT Grant-in-Aid for Scientific Research (A) No.\ 23246071,
and was also partially supported by the National Institute of Information and Communication Technology (NICT), Japan.
The Centre for Quantum Technologies is funded
by the Singapore Ministry of Education and the National Research
Foundation as part of the Research Centres of Excellence programme.
This paper was presented 
in part
at 2014 the International Symposium on Information Theory and Its Applications,
Melbourne, Australia from 26 to 29 October 2014 \cite{Con}.}
\thanks{Masahito Hayashi   is with the Graduate School of Mathematics, Nagoya University,
Furocho, Chikusaku, Nagoya, 464-860, Japan,
and
Centre for Quantum Technologies, National University of Singapore, 3 Science Drive 2, Singapore 117542.
(e-mail: masahito@math.nagoya-u.ac.jp)}
}

\markboth{M. Hayashi: Universal channel coding for general output alphabet}{}

\maketitle

\begin{abstract}

We propose two types of universal codes that are suited to two asymptotic regimes when the output alphabet is possibly continuous. 
The first class has the property that the error probability decays exponentially fast and we identify an explicit lower bound on the error exponent. The other class attains the epsilon-capacity the channel and we also identify the second-order term in the asymptotic expansion. 
The proposed encoder is essentially based on the packing lemma of the method of types.
For the decoder, we first derive a R\'enyi-relative-entropy version of Clarke and
Barron's formula the distance between the true distribution and the
Bayesian mixture, which is of independent interest.
The universal decoder is stated in terms of this formula and quantities used in the
information spectrum method. 
The methods contained herein allow us to analyze universal codes for channels with continuous and discrete output alphabets in a unified manner, and to analyze their performances in terms of the exponential decay of the error probability and the second-order coding rate.
\end{abstract}


%
\IEEEpeerreviewmaketitle

\section{Introduction}\Label{s1}
In wireless communication, 
the channel is described with continuous output alphabet,
e.g., additive white Gaussian noise (AWGN) channel and
Gaussian fading channel.
In these cases,
it is not so easy to identify the channel even though the channel is stationary and memoryless.
Then, it is needed to make a code that achieves good performances for any channel in a set
of multiple stationary and memoryless channels, e.g. a set of MIMO (multiple-input and multiple-output) Gaussian channels (e.g. \cite{A,de4,B,C,D}).
More precisely, it is desired to construct a code that works well for any stationary and
memoryless channel in a given parametric family of possible single-antenna/multi-antenna AWGN
channels for a real wireless communication alphabet.

In the discrete case, to resolve this problem,
Csisz\'{a}r and K\"{o}rner\cite{CK}
proposed universal channel coding 
by employing the method of types.
Since their code construction depends only on the input distribution and the coding rate,
it does not depend on the form of the channel,
which is a remarkable advantage. 
They also provide an explicit form of a lower bound of the exponential decreasing rate of the decoding error probability.
However, their method works only  
when input and output alphabets have finite cardinality.
Hence, their method cannot be applied to any continuous output alphabet
while several practical systems have a continuous output alphabet.
Indeed, even in the continuous output case, 
universal channel codes have been discussed for MIMO Gaussian channels \cite{C,D,E,F},
in which, this problem was often discussed in the framework of compound channel.
Although the studies \cite{Do,D,E} did not cover the general discrete memoryless case,
the paper \cite{F} covered the MIMO Gaussian channels as well as the general discrete memoryless case\footnote{%
While the review paper \cite{C} discussed a universal code in the both cases, their treatments were separated.}.
However, they did not provide any explicit form of a lower bound of the exponential decreasing rate of the decoding error probability.
Therefore, it is desired to invent a universal channel code satisfying the following two conditions.
(1) The universal channel code can be applied to the discrete memoryless case and the continuous case in a unified way. 
(2) The universal channel code has an explicit form of a lower bound of the exponential decreasing rate of the decoding error probability.

Even in the discrete case, Csisz\'{a}r and K\"{o}rner\cite{CK}'s analysis is restricted to the case when 
the transmission rate is strictly smaller than the capacity.
When the transmission rate equals the capacity,
the asymptotic minimum error depends on the second-order coding rate \cite{Strassen,Hsec,Pol}.
For the fixed-length source coding and the uniform random number,
this kind of analysis was done in \cite{H-source}.
Recently, 
with the second order rate,
Polyanskiy \cite{G} addressed a universal code in the framework of compound channel.
Also, Moulin \cite{Moulin2} discussed a universal code in the discrete memoryless case after the conference version of this paper \cite{Con}.
The papers \cite{YN1,YN2,H,I} addressed the optimal second order rate for the mixed channel.
However, no study discussed the universal channel coding with the second-order coding rate for the continuous case although the case of fixed continuous channels was discussed with the higher order analysis including the second-order analysis by Moulin \cite{Moulin}.
Hence, it is also desired to propose a universal channel code working with the second-order coding rate even with the continuous case.
Further, the universal coding with the second order analysis has another problem as follows.
The conventional second order analysis \cite{Strassen,Hsec,Pol} has meaning only when the mutual information is the first order coding rate.
However, the set of such channels has measure zero.
So, such a analysis might be not so useful when we do not know the transition matrix of the channel.

In this paper, we deal with the universal coding with a general output alphabet (including the continuous case)
and derive the exponential decreasing rate of the average error probability
and the second order analysis.
Further, to resolve the above problem for the second order analysis,
we introduce the perturbation with the order 
$O(\frac{1}{n^{\frac{1}{2}}})$ of the channel.
Notice that such a perturbation for distribution is often employed \cite[8.11 Theorem]{Vaart}.
Under this perturbation, we derive the second order analysis, which is the maximum order to achieve the asymptotic constant average error probability.
Indeed, even if the input alphabet is continuous, 
we usually use only a finite subset of the input alphabet for encoding.
Hence, it is sufficient to realize a 
universal channel code for the case
when the input alphabet is finite and the output alphabet is continuous.
In a continuous alphabet,
we need an infinite number of parameters to identify a distribution
when we have no assumption for the true distribution.
In statistics, based on our prior knowledge,
we often assume that the distribution belongs to a certain 
parametric family of distributions as in \cite{F}.
In particular, an exponential family is often employed as a typical example.
This paper adopts this typical assumption as one of our assumptions.
This paper addresses three assumptions.
One is that the output distribution $P_{Y|X=x}$ belongs to 
an exponential family on a general set ${\cal Y}$
for each element $x$ of a given finite set ${\cal X}$.
As explained in Example \ref{ex1},
this assumption covers the usual setting with finite-discrete case
because the set of all distributions on a given finite-discrete set forms
an exponential family.
This assumption also covers the Gaussian fading channel and the multi-antenna Gaussian channel
as addressed in Examples \ref{ex2} and \ref{ex3}.
However, in real wireless communication, the additive noise is not subject to Gaussian distribution \cite{Niranjayan}.
In this case, the channel is not given as an exponential family of channels 
when the fading coefficient is unknown or the distribution of the noise contain unknown parameter,
To cover this case, we consider two general conditions in this paper.
Under these assumptions, 
we provide a universal code with an explicit form of a lower bound of the exponential decreasing rate of the decoding error probability.

To construct our universal encoder,
we employ the method given for the quantum universal channel coding \cite{Ha1},
in which, the packing lemma \cite{CK} 
is employed independently of the output alphabet.
That is, the paper \cite{Ha1} showed that
the encoder given by the packing lemma
can simulate the average of the decoding error probability under the random coding.
Since the method of \cite{Ha1} does not depend on the output alphabet,
it works well with a continuous alphabet.

To construct our universal decoder, 
we focus on the method given for the quantum universal channel coding \cite{Ha1}.
The paper \cite{Ha1} considered a universally approximated output distribution 
by employing the method of types
in the sense of the maximum relative entropy.
Then, it employed the decoder constructed by the information spectrum method 
based on the approximating distributions.
However, in a continuous alphabet, we cannot employ the method of types
so that it is not so easy to give a universally approximated output distribution in the sense of maximum relative entropy.
That is, we cannot directly employ the method in the paper \cite{Ha1}.

To resolve this problem, we focus on Clarke and Barron formula \cite{CB} that shows that
the Baysian average distribution well approximates any independent 
and identical distribution in the sense of relative entropy, i.e., Kullback-Leibler divergence.
Its quantum extension was shown in \cite{Ha2}.
Their original motivation is rooted in universal data compression.
However, they did not discuss the $\alpha$-R\'{e}nyi relative entropy.
In this paper, 
we evaluate the quality of this kind of approximation in terms of the $\alpha$-R\'{e}nyi relative entropy.
Then, we apply the universal decoder constructed by the information spectrum method 
based on the Baysian average distributions.
Modifying the method given in \cite{Ha1}, we derive our lower bound of the error exponent of our universal code.
We also derive the asymptotic error of our universal code with the second-order coding rate.

The remaining part of this paper is the following.
Section \ref{s2} explains notations and our assumptions
In this section, we provide three examples of channels whose output distributions form an exponential family.
Section \ref{s2-2} explains our formulation and obtained results.
In Section \ref{s3}, we give notations for the method of types and our universal encoder.
This part is similar to the previous paper \cite{Ha1}.
In Section \ref{s4},
based on the result by Clarke and Barron \cite{CB},
we derive an $\alpha$-R\'{e}nyi-relative-entropy version of 
Clarke and Barron formula as another new result.
Section \ref{s5} gives our universal decoder and its properties to be applied in the latter sections.
In Section \ref{s6}, we prove our lower bound of our error exponent.
In Section \ref{s8}, we prove the universal achievability for the second order sense.
Appendixes are devoted in several lemmas used in this paper.

\section{Preliminary}\Label{s2}
\subsection{Information quantities}
We focus on an input alphabet ${\cal X}:=\{1, \ldots,d\}$ with finite cardinality
and an output alphabet ${\cal Y}$ that may have infinite cardinality and is a general measurable set.
In this paper, the output alphabet ${\cal Y}$ is treated as a general probability space with a measure $\mu(dy)$
because this description covers the probability space of finite elements and the set of real values.
Hence, when the alphabet ${\cal Y}$ is a discrete set including a finite set,
the measure $\mu(dy)$ is chosen to be the counting measure.
When the alphabet ${\cal Y}$ is a vector space over the real numbers $\RR$,
the measure $\mu(dy)$ is chosen to be the Lebesgue measure.
When we treat a probability distribution $P$ on the alphabet ${\cal Y}$,
it is restricted to a distribution absolutely continuous with respect $\mu(dy)$. 
In the following, we use the lower case $p(y)$ to express the Radon-Nikodym derivative of $P$ with respect to the measure $\mu(dy)$, i.e., the probability density function of $P$ 
so that $P(dy)=p(y) \mu(dy)$.

In general, a channel from ${\cal X}$ to ${\cal Y}$ is described as a collection $\bW$ of conditional
probability measures $W_{x}$ on ${\cal Y}$ for all inputs $x \in {\cal X}$.
Then, we impose the above assumption to $W_{x}$ for any $x \in {\cal X}$.
So, we have $W_x(dy)=w_x(y)\mu(dy)$.
We denote the conditional probability density function by $\bw=(w_x)_{x\in {\cal X}}$.
When a distribution on ${\cal X}$ is given by a probability distribution $P$,
and a conditional distribution on a set ${\cal Y}$ with the condition on ${\cal X}$ is given by $\bV$,
we define the joint distribution $\bW \times P$
on ${\cal X} \times {\cal Y}$ by $\bW \times P(B,x):=W(B|x)P(x)$, 
and the distribution $\bW \cdot P$ on ${\cal Y}$ by $\bW \cdot P(B):=\sum_x W(B|x)P(x)$
for a measurable set $B \subset {\cal Y}$.
Also, we define  the notations $\bw \times P$ and $\bw \cdot P$ as
$\bw \times P(y,x)\mu(dy):=\bW \times P(dy,x)=w_x(y)P(x)\mu(dy)$ and
$\bw \cdot P(y)\mu(dy):=\bW \cdot P(dy)=\sum_{x\in {\cal X}}w_x(y)P(x)\mu(dy)$.
We also employ the notations $W_P:= \bW \cdot P$ and $w_P:= \bw \cdot P$.

We denote the expectation under the distribution $P$ 
by $E_P[~]$. 
Throughout this paper, the base of the logarithm is chosen to be $e$. 
For two distributions $P$ and $Q$ on ${\cal Y}$,
we define the relative entropy
$D(P\|Q):= E_{P}[\log \frac{p(Y)}{q(Y)}]$,
and the value $s D_{1+s}(P\|Q):= \log E_{P}[(\frac{p(Y)}{q(Y)})^{s}]$
for $s >-1 $ when these expectations $E_{P}[\log \frac{p(Y)}{q(Y)}]$ and
$E_{P}[(\frac{p(Y)}{q(Y)})^{s}]$ exist.
The function $s \mapsto s D_{1-s}(P\|Q)$ is a concave function for $s \in [0,1]$
because $-s D_{1-s}(P\|Q)$ can be regarded as a cumulant generating 
function of $- \log \frac{p(Y)}{q(Y)}$.
For $s \in [-1,0) \cup (0,\infty)$, the R\'{e}nyi relative entropy $D_{1+s}(P\|Q)$ is defined as $D_{1+s}(P\|Q):=\frac{s D_{1+s}(P\|Q)}{s}$.
Here, the function $s D_{1+s}(P\|Q)$ is defined for $s=0$, but 
the R\'{e}nyi relative entropy
 $D_{1+s}(P\|Q)$ is not necessarily defined for $s=0$.
In addition, we can define the max relative entropy
\begin{align} 
D_{\max}( P\|Q):= 
\inf \bigg\{a \bigg| a \ge \log \frac{p(y)}{q(y)} 
\hbox{ almost every where with respect to }y 
\hbox{ under the distribution } P\bigg\}
\Label{e5-21-T}.
\end{align} 
So, we have $D_{\max}( P\|Q)= \lim_{\alpha \to \infty}D_\alpha(P\|Q) $.

Given a channel $\bW$ from ${\cal X}$ to ${\cal Y}$ and a distribution $P$ on ${\cal X}$, 
we define the value $s I_{1-s}(P,\bW)$ for $s \in [0,1]$ as
\begin{align}
s I_{1-s}(P,\bW)
:= -  \sup_Q \log \sum_x P(x) e^{-s D_{1-s}(W_x \|Q)}
=\inf_Q s D_{1-s}(\bW \times P\| Q \times P).
\end{align}
Since the minimum of concave functions is a concave function,
the function $s\mapsto s I_{1-s}(P,W)$ is also a concave function.
In fact, H\"{o}lder inequality guarantees that
\begin{align}
s I_{1-s}(P,\bW) =  -(1-s) \log \int
 (\sum_{x}P(x) w_x(y)^{1-s} )^{\frac{1}{1-s}} \mu(dy) 
\Label{eq20}
\end{align}
when the RHS exists \cite{Sibson}\cite[(34)]{q-wire}.
For $s \in (0,1]$,
when $s I_{1-s}(P,\bW)$ is finite, 
$ I_{1-s}(P,\bW)$ is defined as $\frac{s I_{1-s}(P,\bW)}{s}$.
The quantity $I_{1-s}(P,W)$ with (\ref{eq20}) is the same as the Gallager function \cite{Gal} 
with different parametrization for $s$.
We also define the mutual information 
\begin{align}
I(P,\bW):=
\sum_{x}  P(x) D(W_x \| \bW\cdot P)
=
\sum_{x}  P(x) 
\int_{{\cal Y}}
w_x (y) \log \frac{w_x (y)}{\bw\cdot P(y)} \mu(dy)
\end{align}
and its variance
\begin{align}
V(P,\bW)
:=&
\sum_{x}P(x) E_{W_x}
\Big[\log \frac{w_x(Y)}{\bw \cdot P(Y)} - I(P,\bW)\Big]^2 \nonumber \\
=& \sum_{x}P(x) 
\int_{{\cal Y}}
w_x (y)
\Big(\log \frac{w_x(y)}{\bw \cdot P(y)} - I(P,\bW)\Big)^2
\mu(dy).
\end{align}
When the channel satisfies some suitable conditions, we have
\begin{align}
\lim_{s \to 0}I_{1-s}(P,\bW)
=I(P,\bW).\Label{2-26-A}
\end{align}
For example, when the alphabet ${\cal Y}$ has a finite cardinality,
the relation \eqref{2-26-A} holds by choosing the measure $\mu(dy)$ to be the counting measure.

\subsection{Exponential family}
To state the conditions for the main results, in this subsection, we consider a parametric family of distributions
$\{P_{\btheta}\}_{\btheta\in \bTheta }$ on the alphabet ${\cal Y}$
with a finite-dimensional parameter set $\bTheta \subset \mathbb{R}^k$.
Here, we assume that the Radon-Nikodym derivative $p_{\btheta}(y)$ is differentiable with respect to $\theta$.
To cover a class of proper channels, we introduce 
an exponential family of distributions, which covers so many useful distributions in statistics,
and has been widely recognized as a key concept of statistics \cite{Lehmann2,AN}.
A parametric family of distributions $\{P_{\btheta}\}_{\btheta\in \bTheta }$
on the measurable set ${\cal Y}$
is called an {\it exponential family}
with a parametric space $\bTheta \subset \mathbb{R}^k$ 
when a distribution $P_0$ is absolutely continuous with respect to a measure $\mu(dy)$
on a measurable set ${\cal Y}$
and the parametric family of distributions $\{P_{\btheta}\}_{\btheta\in \bTheta }$ is written as \cite{AN}
\begin{align}
p_{\btheta}(y)=
p_{0}(y) e^{\sum_{j=1}^{k} \theta^j g_{j}(y)- \phi(\btheta)}
\end{align}
with generators $g_{j}(y)$ and satisfies the following conditions.
\begin{description}
\item[A1]
The potential function $\phi(\btheta)$ equals 
the cumulant generating function of $g_{j}$, i.e.,
\begin{align}
e^{\phi(\btheta)}= 
\int_{{\cal Y}} p_{0}(y) e^{\sum_{j=1}^{k} \theta^j g_{j}(y)}
\mu(dy)< \infty
\end{align}
for $\btheta \in \bTheta$.

\item[A2]
$\phi$ is a $C^2$ function on $\bTheta$, i.e.,
the Hessian matrix $J_{\btheta}=(J_{\btheta|i,j})_{i,j}$ of $\phi$ is continuous on $\bTheta$.
\end{description}
The function $\psi$ is called the potential.
In the above assumption,
the Hessian matrix $J_{\btheta}$ equals the covariance matrix of the random variables $(g_{j}(Y))_i$
under the distribution $P_{\btheta} $.
In statistics, the Hessian matrix $J_{\btheta}$ is called Fisher information matrix,
and expresses the bound of the precision of the estimation of the parameter $\btheta$.

\begin{example}[Multinomial distributions]\Label{Mul}
The set of multinomial distributions in a finite output set ${\cal Y}=\{0,1, \ldots, m\}$
forms an exponential family.
We choose the generator $g_j(y):= \delta_{j,y}$ for $j=1, \ldots,m$
and define $p_0(y):= \frac{1}{m+1}$. 
Remember that $\mu(dy)$ is given as the counting measure.
For parameters $\theta=(\theta^{j})$ with $ j=1, \ldots,m$,
the distributions $p_{\btheta}(y)$ is given as 
\begin{align}
p_{\btheta}(y):=
\left\{
\begin{array}{ll}
\frac{e^{\theta^{y}}}{1+\sum_{j=1}^{m} e^{\theta^{j}}}
& \hbox{ when } y \ge 1 \\
\frac{1}{1+\sum_{j=1}^{m} e^{\theta^{j}}}
& \hbox{ when } y =0 .
\end{array}
\right.
\end{align}
The potential $\phi(\btheta)$ is $\log (1+\sum_{j=1}^{m} e^{\theta^{j}})-\log(m+1)$.
\end{example}

\begin{example}[Poisson distributions]\Label{Poi}
The set of Poisson distributions is a one-parameter exponential family
on the set of natural numbers $\NN:=\{0,1, \ldots\}$.
We set $p_0(y):=\frac{e^{-1}}{n !}$ and choose the generator $g(y):= y $.
Then, we have 
$p_{\theta|Poi}(y):=\frac{e^{n \theta-e^\theta}}{n!}$
with the potential $\phi(\theta)= e^\theta -1 $.
\end{example}

\begin{example}[Gaussian distributions]\Label{Gau}
The set of Gaussian distributions forms a one-parameter exponential family on $\RR$.
We set $p_0(y):=\frac{e^{-\frac{(y)^2}{2}}}{\sqrt{2\pi}}$ and choose the generator $g(y):= y $.
Then, we have 
$p_{\theta}(y):=\frac{e^{-\frac{(y-\theta)^2}{2}}}{\sqrt{2\pi}}$
with the potential $\phi(\theta)= \frac{(\theta)^2}{2} $.
\end{example}

Now, we consider a parametric family of channels 
from ${\cal X}$ to ${\cal Y}$ with a finite-dimensional parameter set 
The channel is parameterized as $\{\bW_{\btheta}\}_{\btheta \in \bTheta}$
with a finite-dimensional parameter set $\bTheta \subset \mathbb{R}^{k}$.

Now, using the conditions  for an exponential family, we make a condition for a family for the channels $\{\bW_{\btheta}\}_{\btheta\in \bTheta }$.
A parametric family of channels $\{\bW_{\btheta}\}_{\btheta\in\bTheta}$ is called 
an {\it exponential family of channels} with a parametric space $\bTheta \subset \mathbb{R}^k$ 
when a channel $\bW_{0}=(W_{0,x})$ from ${\cal X}$ to ${\cal Y}$ is composed of 
absolutely continuous distributions $W_{0,x}$ with respect to a measure $\mu(dy)$
and the parametric family of channels $\{\bW_{\btheta}\}_{\btheta\in\bTheta}$ is written as 
\begin{align}
w_{\btheta,x}(y)=
w_{0,x}(y) e^{\sum_{j=1}^{k} \theta^{j} g_{j,x}(y)- \phi_x(\btheta)}
\end{align}
for any $x \in {\cal X}$ with generators $g_{j,x}(y)$ 
and satisfies the following conditions.
\begin{description}
\item[B1]
The potential function $\phi_x(\btheta)$ equals 
the cumulant generating function of $g_{j,x}$, i.e.,
\begin{align}
e^{\phi_x(\btheta)}= 
\int_{{\cal Y}} w_{0,x}(y) e^{\sum_{j=1}^{k} \theta^j g_{j,x}(y)}
\mu(dy)< \infty
\end{align}
for $\btheta \in \bTheta$.

\item[B2]
$\phi_x$ is a $C^2$ function on $\bTheta$, i.e.,
the Hessian matrix $J_{\btheta,x}=(J_{\btheta,x|i,j})_{i,j}$ of $\phi$ is continuous on $\bTheta$.

\end{description}
When ${\cal Y}$ equals ${\cal X}$, this definition of exponential family 
has been discussed in many papers in the context of Markovian processes
\cite{Bhat1,Bhat2,Hudson,KM,Feigin},
and is different from the definition in the papers \cite{NaKa,Naga,HW,WH}.
Here, for convenience, we say that a family of channels satisfies Condition A
when all of the above conditions hold, i.e., it is an exponential family of channels.


Here, we list typical examples as follows.
All of the below examples satisfy Condition B.

\begin{example}[Finite set]\Label{ex1}
Consider a finite input set ${\cal X}=\{1, \ldots, d\}$ and
a finite output set ${\cal Y}=\{0,1, \ldots, m\}$.
Then, the measure $\mu(dy)$ is chosen to be the counting measure. 
We choose the generators $g_{(i,j),x}(y):=\delta_{i,x}\delta_{j,y}$ with $i=1, \ldots, d, j=1, \ldots,m$.
For parameters $\btheta=(\theta^{i,j})$,
the exponential family of channels is given as
\begin{align}
\bw_{\btheta,x}(y)
:= 
\left\{
\begin{array}{ll}
\frac{e^{\theta^{x,y}}}{1+\sum_{j=1}^{m} e^{\theta^{x,j}}}
& \hbox{ when } y \ge 1 \\
\frac{1}{1+\sum_{j=1}^{m} e^{\theta^{x,j}}}
& \hbox{ when } y =0 .
\end{array}
\right.
\end{align}
Then, the set of output distributions
$\{\bW_{\btheta}\}_{\btheta}$ forms an exponential family for $x \in {\cal X}$.
\hfill $\square$\end{example}

\begin{example}[Gaussian fading channel]\Label{ex2}
Assume that the output set ${\cal Y}$ is the set of real numbers $\mathbb{R}$
and the input set ${\cal L}$ is a finite set $\{1, \ldots, d\}$.
We choose $d$ elements $x_1, \ldots, x_d \in \mathbb{R}$ as input signals.
The additive noise $Z$ is assumed to be subject to 
the Gaussian distribution with the expectation $b$ and the variance $v$.
Then, we assume that the received signal $Y$ is given by the scale parameter $a$ as
\begin{align}
Y= a x_i +Z.\Label{E5-18}
\end{align}
The class of these channels is known as Gaussian fading channels, and its compound channel is discussed in the paper \cite{E}.
We choose the generators 
$g_{1,\ell}(y):=-\frac{y^2}{2}$, 
$g_{2,\ell}(y):=y x_\ell$, and 
$g_{3,\ell}(y):=y$. 
By using the three parameters $\theta^1:=\frac{1}{v}$, $\theta^2:=\frac{a}{v}$,
and $\theta^3:=\frac{b}{v}$,
the channels form an exponential family of channels as
\begin{align}
& w_{\btheta,\ell}(y)
:=
\frac{1}{\sqrt{2\pi v}} 
e^{-\frac{(y-a x_\ell-b)^2}{2v} } \nonumber \\
=&
\frac{\theta^1}{\sqrt{2\pi}} 
e^{-\frac{-y^2}{2} \theta^1+ y (\theta^2 x_\ell+\theta^3) 
-\frac{(\theta^2 x_\ell+\theta^3)^2}{2 \theta^1}}.
\end{align}
\hfill $\square$\end{example}

\begin{example}[multi-antenna Gaussian channel]\Label{ex3}
We consider the constant multi-antenna (MIMO) Gaussian channel
when the sender has $t$ antennas and the receiver has $r$ antennas.
Then, the output set ${\cal Y}$ is given as the set of $r$-dimensional real numbers $\mathbb{R}^r$
and the input set ${\cal L}$ is given as a finite set $\{1, \ldots, d\}$.
In this case, the sender chooses 
$d$ elements $\vec{x}_1, \ldots, \vec{x}_d \in \mathbb{R}^t$ 
as input signals.
The $r$-dimensional additive noise $\vec{Z}$ is assumed to be subject to 
the $r$-dimensional Gaussian distribution with the expectation $\vec{b}\in \mathbb{R}^d$ and the covariance matrix $v_{i,i'}$.
When the input is $\vec{x}_{\ell}$,
we assume that the received signal $\vec{Y}$ is written by an $r \times t$ matrix $(a_{i,j})$ as
\begin{align}
\vec{Y}= \sum_{j=1}^t a_{i,j} x^j_{\ell} +\vec{Z}.
\end{align}
The class of these channels is known as constant multi-antenna Gaussian channels \cite{de4}.
We choose three kinds of generators 
$g_{(1,i,i'),\ell'}(y):=-\frac{y^i y^{i'}}{2}$, 
$g_{(2,i,j),\ell'}(y):=y^i x^{j}_{\ell'}$, and $g_{(3,i),\ell'}(y):=y^i$. 
Then, we define three kinds of parameters:
The symmetric matrix $(\theta^{1,i,i'})_{i,i'}$ is defined as the inverse matrix of the covariance matrix $v_{i,i'}$.
The matrix $(\theta^{2,i,i'})_{i,i'} $ is defined as
$\theta^{2,i,j}:=\sum_{i'} \theta^{1,i,i'} a_{i',j}$.
The vector $\theta^{3,i}:=\sum_{i'} \theta^{1,i,i'} b^{i'}$.
So, the channels form an exponential family of channels as
\begin{align}
& w_{\btheta,\ell}(\vec{y})
:=
\frac{ \det (\theta^1)}{\sqrt{2\pi}} 
e^{-\frac{1}{2} 
\sum_{i,i'} \theta^{1,i,i'} 
(y^i- \sum_j a_{i,j} x^j_\ell-b^i)
(y^{i'}- \sum_{j'} a_{i',j'} x^{j'}_{\ell}-b^{i'})
} \nonumber \\
=&
\frac{ \det (\theta^1)}{\sqrt{2\pi}} 
e^{-\frac{1}{2} 
\sum_{i,i'} \theta^{1,i,i'} y^i y^{i'}
+
\sum_{i} y^{i} (\theta^{3,i}+ \sum_{j}\theta^{2,i,j} x^j_\ell )
-\frac{1}{2} 
\sum_{i,i'} 
((\theta^{1})^{-1})^{i,i'} 
(\theta^{3,i}+ \sum_{j}\theta^{2,i,j} x^j_\ell )
(\theta^{3,i'}+ \sum_{j}\theta^{2,i',j} x^j_\ell )}.
\end{align}
\hfill $\square$\end{example}

\begin{example}[Poisson channel]\Label{ex4}
When the signal is too weak, we have Poisson channel as follows.
Assume that the output set ${\cal Y}$ is $\NN$
and the input set ${\cal X}$ is a finite set $\{1, \ldots, d\}$.
Then, the measure $\mu(dy)$ is chosen to be the counting measure. 
We choose three kinds of generators $g_{i,x}(y):=\delta_{i,x}y$. 
By using the parameter $\btheta=(\theta_i)_{i=1}^d$,
the Poisson channel is given as an exponential family of channel with
\begin{align}
& w_{\btheta,x}(y)
:=P_{\theta_x|Poi}(y).
\end{align}
For the definition of $P_{\theta_x|Poi}$, see Example \ref{Poi}.
\hfill $\square$\end{example}

Here, we summarize the number in Table \ref{T1}.

\begin{table}[htpb]
  \caption{Numbers used in this paper}
\Label{T1}
\begin{center}
  \begin{tabular}{|c|l|} 
\hline
$k$ & Dimension of parameter \\
\hline
$d$ & Cardinality of input alphabet \\
\hline
$n$ & No. of use of channels \\
\hline
  \end{tabular}
\end{center}
\end{table}

\subsection{General assumptions}
In real wireless communication, the additive noise is not subject to Gaussian distribution \cite{Niranjayan}.
Now, we consider the case when 
the additive noise $Z$ in \eqref{E5-18} is subject to a general parametric family of distributions 
$\{P_{\btheta}\}_{\btheta}$ on ${\cal Y}=\RR$, which is not a Gaussian distribution.
Then, the probability density function of the output distributions are given as 
$ w_{\btheta,\ell}(y)=p_{\btheta}(y-ax_\ell)$, 
and do not form an exponential family in general.
Further, even when the distribution $P_{\btheta}$ of the additive noise $Z$ is known,
when the fading coefficient $a$ is unknown, 
the set of our channels does not form an exponential family in general.
 Hence, it is needed to relax the condition of exponential family for practical purpose.
As a preparation, we consider the following condition for a family of distributions $\{P_{\btheta}\}_{\btheta\in \bTheta}$ on the measurable set ${\cal Y}$ with $\bTheta \subset \RR^d$.


\begin{description}
\item[C1]
The parametric space $\bTheta$ is compact.

\item[C2]
The likelihood ratio derivative 
$l_{\btheta|i}(y):= \frac{\partial}{\partial \theta^i}\log p_{\btheta}(y)$ exists for $i=1, \ldots,k$
even on the boundary of $\bTheta$.

\item[C3]
The Fisher information matrix
$ J_{\btheta|i,j}
:=\int_{{\cal Y}} p_{\btheta}(y) l_{\btheta|i}(y) l_{\btheta|j}(y) \mu(dy)$ can be defined.
Also, the map $\btheta \mapsto  J_{\btheta|i,j}$ is continuous for $i,j=1, \ldots, d$.

\item[C4]
For any $s >0$, 
the map $(\btheta,\btheta')\mapsto D_{1+s}(P_{\btheta}\| P_{\btheta'})$
is continuous.
Also, when $s>0$ is fixed, the convergence
$2\lim_{\epsilon\to 0}\frac{D_{1+s}(P_{\btheta}\| P_{\btheta+\bxi \epsilon})}{\epsilon^2}
=\sum_{i,j}\bxi_j \bxi_i J_{\btheta|i,j}$
holds 
and is uniform for $\btheta$ and $\bxi \in \{\bxi \in \RR^d|  \|\bxi\|=1\}$.
\end{description}

We say that a family of distributions satisfies Condition C when all of the above conditions hold.
For example, when an exponential family has a compact parameter space,
the likelihood ratio derivative $l_{\btheta|i}(y)$ is given as $g_{j}(y)-E_{P_{\btheta}}[g_{j}(Y)]$ so
that it satisfies Condition C.
Under this condition, to prove several required properties, we can effectively employ the compactness of the parameter space $\bTheta$.
For example,
when a probability density function $P$ is differentiable on ${\cal Y}=\RR$,
the support $P$ is ${\cal Y}=\RR$,
and the integral satisfies the condition
\begin{align}
\int_{-\infty}^{\infty} (\frac{d p(y)}{d y}/p(y) )^2 p(y) d y < \infty \Label{e5-21}
\end{align}
the family of distribution $p_{\btheta}(y):= p(y-\btheta)$
satisfies Condition C \cite{Akahira}.

Next, we extend Condition C to a family of channels
$\bW_{\btheta}=(W_{\btheta,x})$ from a discrete alphabet ${\cal X}$ to a general alphabet ${\cal Y}$
with the parameter set $\bTheta$. 
\begin{description}
\item[D1]
The parametric space $\bTheta$ is compact.

\item[D2]
The likelihood ratio derivative 
$l_{\btheta,x|i}(y):= \frac{\partial}{\partial \theta^i}\log w_{\btheta,x}(y)$ 
exists for $i=1, \ldots,k$ and $x \in {\cal X}$.

\item[D3]
The Fisher information matrix
$ J_{\btheta,x|i,j}
:=\int_{{\cal Y}} w_{\btheta,x}(y) l_{\btheta,x|i}(y) l_{\btheta,x|j}(y) \mu(dy)$ can be defined for
$x \in {\cal X}$.
Also, the map $\btheta \mapsto  J_{\btheta,x|i,j}$ is continuous for $i,j=1, \ldots, d$ and $x \in {\cal X}$.

\item[D4]
For any $s >0$ and $x \in {\cal X}$, 
the map $(\btheta,\btheta')\mapsto D_{1+s}(W_{\btheta,x}\| W_{\btheta,x'})$
is continuous.
Also, when $s>0$ and $x \in {\cal X}$ are fixed, the convergence
$2\lim_{\epsilon\to 0}\frac{D_{1+s}(W_{\btheta,x}\| W_{\btheta,x+\bxi \epsilon})}{\epsilon^2}
=\sum_{i,j}\bxi_j \bxi_i J_{\btheta,x|i,j}$
holds 
and is uniform for $\btheta$ and $\bxi \in \{\bxi \in \RR^{k}|  \|\bxi\|=1\}$.
\end{description}
Hence, we say that a family of channels satisfies Condition D
when all of the above conditions hold.
For example, when an exponential family of channels has a compact parameter space,
it satisfies Condition D.
For a given distribution $P$ on ${\cal Y}=\RR$ for the additive noise,
we consider the channel $ \bW_{\btheta}(y|i)=P(y- \theta_i)$ with a compact space $\bTheta$.
Even when the additive noise in a wireless communication is not subject to Gaussian distribution, 
when the distribution $P$ satisfies the condition \eqref{e5-21},
the family of channels $ \{\bW_{\btheta}\}_{\btheta}$ satisfies Condition D.

However, it is better to remove the assumption of the compactness.
For this purpose, we introduce the conditions E1 and E4 instead of C1 and C4
 for a family of distributions
$\{P_{\btheta}\}_{\btheta\in \bTheta}$ on the measurable set ${\cal Y}$ with $\bTheta \subset \RR^k$.

\begin{description}
\item[E1]
There exists a sequence of compact subsets $\{\bTheta^i\}_{i=1}^{\infty}$ 
satisfying the following conditions.
(1) $\bTheta^i \subset \bTheta^{i+1}$.
(2) For any $\btheta \in \bTheta$, there exists $\bTheta^i$ such that
$\btheta \in \bTheta^i$.
\item[E4]
The uniformity of the convergence in C4 holds on any compact subspace $\bTheta^i$. 
\end{description}
We say that a family of channels satisfies Condition E
when the conditions E1 and E4 hold as well as C2 and C3.

Next, we extend Condition E to a family of channels
$\bW_{\btheta}=(W_{\btheta,x})$ from a discrete alphabet ${\cal X}$ to a general alphabet ${\cal Y}$
with the parameter set $\bTheta$. 
Instead of D1 and D4, we introduce the conditions F1 and F4.

\begin{description}
\item[F1]
For each $x \in {\cal X}$, 
there exists a sequence of compact subsets $\{\bTheta^i\}_{i=1}^{\infty}$ 
satisfying the following conditions.
(1) $\bTheta^i \subset \bTheta^{i+1}$.
(2) For any $\btheta \in \bTheta$, there exists $\bTheta^i$ such that
$\btheta \in \bTheta^i$.
\item[F4]
The uniformity of the convergence in D4 holds on any compact subspace $\bTheta^i$. 
\end{description}
We say that a family of channels satisfies Condition F
when the conditions F1 and F4 hold as well as D2 and D3.
Since Condition F does not require compactness,
any exponential family of channels satisfies Condition F.

\begin{example}\Label{Ex8}
Assume that the output set ${\cal Y}$ is the set of real numbers $\mathbb{R}$
and the input set ${\cal L}$ is a finite set $\{1, \ldots, d\}$.
We choose $d$ elements $x_1, \ldots, x_d \in \mathbb{R}$ as input signals.
We assume that the additive noise $Z$ in \eqref{E5-18} is subject to a general parametric family of distributions 
$\{P_{\btheta}\}_{\btheta}$ on ${\cal Y}=\RR$ satisfying Condition E.
Our family of channels $\{\bW_{\btheta}\}_{\btheta}$ is given as 
$ w_{\btheta,\ell}(y)=p_{\btheta}(y-ax_\ell)$. 
Since $\{P_{\btheta}\}_{\btheta}$ on ${\cal Y}=\RR$ satisfies Condition E,
this family of channels satisfies Condition F.
\end{example}

\begin{example}\Label{Ex9}
Assume that the output set ${\cal Y}$ is the set of real numbers $\mathbb{R}$
and the input set ${\cal L}$ is a finite set $\{1, \ldots, d\}$.
We choose $d$ elements $x_1, \ldots, x_d \in \mathbb{R}$ as input signals.
We assume that the additive noise $Z$ in \eqref{E5-18} is subject to a distribution $P$ satisfying 
\eqref{e5-21}.
However, we do not know the fading coefficient $\theta$.
Our family of channels $\{\bW_{\theta}\}_{\theta \in \RR}$ is given as 
$ w_{\theta,\ell}(y)=p(y-\theta x_\ell)$. 
Due to the condition \eqref{e5-21},
this family of channels satisfies Condition F.
\end{example}

\section{Main results}\Label{s2-2}
\subsection{Exponential evaluation}
In this paper, we address the $n$-fold stationary memoryless channel of $\bW_{\btheta}$, i.e.,
we focus on the channel $\bW_{\btheta}^n$ whose probability density function is defined as 
$w_{\btheta,x^n}^n(y^n):=\prod_{i=1}^n w_{\btheta,x_i}(y_i)$
with $x^n:=(x_1,\ldots,x_n)\in {\cal X}^n$
and $y^n:=(y_1,\ldots,y_n)\in {\cal Y}^n$.
When the set of messages is ${\cal M}_n:=
\{1, \ldots, M_n\}$,
the encoder is given as a map 
$E_n$ from ${\cal M}_n $ to ${\cal X}^n$,
and the decoder is given as a map 
$D_n$ from ${\cal Y}^n$ to ${\cal M}_n $.
The triple $\Phi_n:=(M_n,E_n,D_n)$
is called a code, and 
the size $M_n$ is often written as $|\Phi_n|$.
The decoding error probability is 
$e_{\btheta}(\Phi_n):= \frac{1}{M_n}\sum_{i=1}^{M_n}\sum_{j \neq i} W_{\btheta,E_n(i)}^n( D_n^{-1}(j))$.

\begin{theorem}\Label{Th1}
Given real numbers $R$ and $R_1$ with $R_1>R$,
a distribution $P$ on ${\cal X}$,
and a family of channels $\{\bW_{\btheta}\}_{\btheta}$,
we assume that 
the family of channels $\{\bW_{\btheta}\}_{\btheta}$ satisfies Condition B, D, or F,
and that $s I_{1-s}(P,\bW_{\btheta} )$ can be defined for any $s \in [0,1]$.
Then, there exists a sequence of codes $\Phi_n$ 
with the size $|\Phi_n|=e^{nR}$ 
satisfying that
\begin{align}
\lim_{n \to \infty}\frac{-1}{n} 
\log
e_{\btheta} (\Phi_n)
\ge 
\min(\max_{s \in [0,1]}(s I_{1-s}(P,\bW_{\btheta} ) -s R_1), R_1-R)
\Label{eq2} 
\end{align}
for every $\btheta \in \bTheta$.
\hfill $\square$\end{theorem}

The proof of Theorem \ref{Th1} is given as the combination of Sections \ref{s3}, \ref{s5}, and  \ref{s6}.
Section \ref{s3} gives the encoder and
Section \ref{s5} gives the decoder and a part of the evaluation of the decoding error probability.
Section \ref{s6} evaluates the decoding error probability by using several formulas given in 
Sections \ref{s3} and \ref{s5}.
While $R$ expresses the coding rate, our code needs another rate $R_1>R$, which decides the form of the universal decoder.
This is because our decoder employs the likelihood ratio test of a certain simple hypothesis testing,
and the rate $R_1$ describes the likelihood ratio.
Hence, it is crucial to decide the rate $R_1$ for the construction of our universal code.

For simplicity, we consider the case when we can identify the parameter $\btheta$, we can choose $\bTheta_0$ to be $\{\btheta\}$.
In this case, the best choice of $R_1$ is the real number $R_1$ satisfying 
$\max_{s \in [0,1]}(s I_{1-s}(P,\bW_{\btheta} ) -s R_1)= R_1-R$,
which is given as 
$R_1=R+\max_{s \in [0,1]}\frac{1}{1+s}(s I_{1-s}(P,\bW_{\btheta} ) -s R)$,
which can be seen as a special case of Lemma \ref{L11-27-1} later.
Hence, our lower bound of the exponent is 
$\max_P  \max_{s\in [0,1]}\frac{1}{1+s}(s I_{1-s}(P,\bW_{\btheta} ) -s R) $.
This value is strictly smaller than Gallager's exponent
$\max_{s \in [0,1]}\frac{1}{1-s}(s I_{1-s}(P,\bW_{\btheta} ) -s R)$\cite{Gal}.

Now, we consider the general case when we assume that the true channel parameter $\btheta$ belongs to a subset $\bTheta_0 \subset \bTheta$.
Then, we propose the following method to choose the rate $R_1$ and
the distribution $P$ to identify our code for a given transmission rate 
$R< \sup_P \inf_{\btheta \in \bTheta_0} I(P, \bW_{\btheta})$.

\begin{description}
\item[(M1)]
When $\argmax_P \inf_{\btheta \in \bTheta_0} I(P, \bW_{\btheta})$ 
is a non-empty set, we fix an element $P_1$ in this set. 
Otherwise, we choose a distribution $P_1$ such that
$\inf_{\btheta \in \bTheta_0} I(P_1, \bW_{\btheta})$ is sufficiently close to 
$\sup_P \inf_{\btheta \in \bTheta_0} I(P, \bW_{\btheta}))$.
Next, we choose $R_1 \in (R, \inf_{\btheta \in \bTheta_0} I(P_1, \bW_{\btheta}))$.
Then, the exponential decreasing rate is greater than
$\min(\max_{s\in [0,1]}(s I_{1-s}(P,\bW_{\btheta} ) -s R_1) ,R_1-R)>0$
when the true parameter is $\btheta \in \bTheta_0$.
\end{description}

Now, we consider another method to improve the bound 
$\min(\max_{s\in [0,1]}(s I_{1-s}(P,\bW_{\btheta} ) -s R_1) ,R_1-R)$.
For this purpose, we prepare the following lemma whose proof of Lemma \ref{L11-27-1} is given in Appendix \ref{A1}.

\begin{lemma}\Label{L11-27-1}
When the function $s \mapsto s I_{1-s}(P,\bW_{\btheta} )$
is a $C^1$ function for any $\btheta \in \bTheta$, we have
\begin{align}
& \max_{R_1} \inf_{\btheta\in \bTheta_0}\min (\max_{s \in [0,1]}
(s I_{1-s}(P,\bW_{\btheta} ) -s R_1),R_1-R) \nonumber \\
&=
\inf_{\btheta\in \bTheta_0}\max_{R_1} \min 
(\max_{s \in [0,1]}( s I_{1-s}(P,\bW_{\btheta} ) -s R_1),R_1-R)\nonumber  \\
&=\inf_{\btheta\in \bTheta_0} 
\max_{s\in [0,1]}\frac{1}{1+s}(s I_{1-s}(P,\bW_{\btheta} ) -s R) .\Label{eq11-27-1}
\end{align}
The maximum value $\max_{R_1} \min (\max_{s \in [0,1]}
(s I_{1-s}(P,\bW_{\btheta} ) -s R_1),R_1-R)$ is attained when 
$R_1= R+\max_{s\in [0,1]}\frac{1}{1+s}(s I_{1-s}(P,\bW_{\btheta} ) -s R)$.
In particular, the maximum value\par 
\noindent$\max_{R_1} \inf_{\btheta\in \bTheta_0}\min (\max_{s \in [0,1]}(s I_{1-s}(P,\bW_{\btheta} ) -s R_1),R_1-R)$ is attained when 
$R_1= R+$\par 
\noindent$\inf_{\btheta\in \bTheta_0}\max_{s\in [0,1]}\frac{1}{1+s}(s I_{1-s}(P,\bW_{\btheta} ) -s R)$.
\hfill $\square$\end{lemma}

Based on Lemma \ref{L11-27-1},we propose the following method to choose the rate $R_1$ and
the distribution $P$, which improve the above method.

\begin{description}
\item[(M2)]
When $\argmax_P 
\inf_{\btheta\in \bTheta_0} 
\max_{s\in [0,1]}\frac{1}{1+s}(s I_{1-s}(P,\bW_{\btheta} ) -s R)$ is a non-empty set, 
we fix an element $P_2$ in this set. 
Otherwise, we choose a distribution $P_2$ such that
$\inf_{\btheta\in \bTheta_0} 
\max_{s\in [0,1]}\frac{1}{1+s}(s I_{1-s}(P_2,\bW_{\btheta} ) -s R)$ 
is arbitrarily close to 
$\sup_P \inf_{\btheta\in \bTheta_0} 
\max_{s\in [0,1]}\frac{1}{1+s}(s I_{1-s}(P,\bW_{\btheta} ) -s R)$.
Then, we choose $R_1^*:=
\argmax_{R_1} \inf_{\btheta\in \bTheta_0}\min (\max_{s \in [0,1]}
(s I_{1-s}(P_2,\bW_{\btheta} ) -s R_1),R_1-R)$.
\end{description}

For simplicity, we consider the case when 
$\max_P \inf_{\btheta\in \bTheta_0} 
\max_{s\in [0,1]}\frac{1}{1+s}(s I_{1-s}(P,\bW_{\btheta} ) -s R)$ exists.
The worst case of the lower bound of the exponential decreasing rate of (M2) is
$\inf_{\btheta \in \bTheta_0}\min(\max_{s\in [0,1]}(s I_{1-s}(P_2,\bW_{\btheta} ) -s R_1^*) ,R_1^*-R)$.
Due to \eqref{eq11-27-1} in Lemma \ref{L11-27-1},
the lower bound is calculated as
\begin{align}
&\inf_{\btheta \in \bTheta_0}\min(\max_{s\in [0,1]}(s I_{1-s}(P_2,\bW_{\btheta} ) -s R_1^*) ,R_1^*-R)
\nonumber \\
=&
\max_{R_1} \inf_{\btheta\in \bTheta_0}\min (\max_{s \in [0,1]}
(s I_{1-s}(P_2,\bW_{\btheta} ) -s R_1),R_1-R)
\nonumber \\
=&
\inf_{\btheta\in \bTheta_0} 
\max_{s\in [0,1]}\frac{1}{1+s}(s I_{1-s}(P_2,\bW_{\btheta} ) -s R) \nonumber \\
=&
\max_P \inf_{\btheta\in \bTheta_0} 
\max_{s\in [0,1]}\frac{1}{1+s}(s I_{1-s}(P,\bW_{\btheta} ) -s R) 
\ge 
\inf_{\btheta\in \bTheta_0} 
\max_{s\in [0,1]}\frac{1}{1+s}(s I_{1-s}(P_2,\bW_{\btheta} ) -s R) .
\end{align}
Since 
$\inf_{\btheta\in \bTheta_0} 
\max_{s\in [0,1]}\frac{1}{1+s}(s I_{1-s}(P_2,\bW_{\btheta} ) -s R) $
expresses the worst case of the lower bound of the exponential decreasing rate of (M1), 
(M2) provides a better lower bound than (M1).

\subsection{Second order evaluation}
When we need to realize the transmission rate close to the capacity,
we need to employ the second order analysis.

\begin{theorem}\Label{Th2}
Given real numbers $R_1^*$ and $R_2^*$,
a distribution $P$ on ${\cal X}$,
and a family of channels $\{\bW_{\btheta}\}_{\btheta}$,
we assume that the family of channels $\{\bW_{\btheta}\}_{\btheta}$ satisfies Condition A, B, or C,
and that $0<I(P,\bW_{\btheta} )< \infty$ and $0<V(P,\bW_{\btheta})<\infty$ for
any parameter $\btheta\in \bTheta$.
Then, there exists a sequence of codes $\Phi_n$ 
with the size $|\Phi_n|=e^{nR_1^*+\sqrt{n} R_2^*- n^{\frac{1}{4}}}$ 
satisfying the following conditions.
(1) 
When $I(P,\bW_{\btheta} ) > R_1^*$,
we have \begin{align}
\lim_{n \to \infty}
e_{\btheta} (\Phi_n)=0.\Label{11-27-6b}
\end{align}
(2) 
When $I(P,\bW_{\btheta} )=R_1^*$,
we have 
\begin{align}
\lim_{n \to \infty}
e_{\btheta} (\Phi_n)
\le
\int_{-\infty}^{\frac{R_2^*}{\sqrt{V(P,\bW_{\btheta} )}}}
 \frac{1}{\sqrt{2\pi}}
\exp( - \frac{x^2}{2}) dx
\Label{eq2-2} .
\end{align}
\hfill $\square$\end{theorem}
The proof of Theorem \ref{Th2} is given in Section \ref{s8}.
Therefore, in order that the upper bound is smaller than $\frac{1}{2}$, i.e., 
$\int_{-\infty}^{\frac{R_2^*}{\sqrt{V(P,\bW_{\btheta} )}}}
 \frac{1}{\sqrt{2\pi}}
\exp( - \frac{x^2}{2}) dx < \frac{1}{2}$, 
$R_2^*$ needs to be a negative value.
For example, when $\bTheta=\{\theta\}$,
we choose $P$ as an element of
$\{P' | I(P,\bW_{\btheta})=\max_P I(P,\bW_{\btheta}),
V(P,\bW_{\btheta})=V_-\}$
with 
$V_-:=\min_{P } 
\{V(P,\bW_{\btheta})|I(P,\bW_{\btheta})=\max_P I(P,\bW_{\btheta})\}.$
Then, we can realize the minimum error 
$ 
\int_{-\infty}^{\frac{R_2^*}{\sqrt{V_-}}}
 \frac{1}{\sqrt{2\pi}}
\exp( - \frac{x^2}{2}) dx$
with the coding length $n \max I(P,\bW_{\btheta})+\sqrt{n}R_2^* $ when $R_2^*\le 0$ 
\cite{Strassen,Hsec,Pol}.

However, most of channels $W_{\btheta}$ do not satisfy the condition $I(P,\bW_{\btheta} )=R_1^*$
for a specific first order rate $R_1^*$.
That is, Theorem \ref{Th2} gives the asymptotic average error probability $0$ or $1$
as the coding result for most of channels $W_{\btheta}$.
This argument does not reflect the real situation properly because 
many channels have their average error probability between $0$ and $1$, i.e., 
$0$ nor $1$ is not an actual value of average error probability.
To overcome this problem, 
we introduce the second order parameterization $\btheta_1 +\frac{1}{\sqrt{n}}\btheta_2\in \bTheta$.
This parametrization is applicable when the unknown parameter belongs to the neighborhood of $\btheta_1$.
Given a parameter $\btheta$, we choose 
the parameter $\btheta_2:= (\btheta-\btheta_1)\sqrt{n}$.
This parametrization is very conventional in statistics.
For example, in statistical hypothesis testing, 
when we know that the true parameter $\btheta$ belongs to the neighborhood of $\btheta_1$,
using the new parameter $\btheta_2$,
we approximate the distribution family by the Gaussian distribution family \cite{Lehmann}.
In this sense, the $\chi^2$-test gives the asymptotic optimal performance. 
As another example, this kind of parametrization is employed to discuss local minimax theorem \cite[8.11 Theorem]{Vaart}.

The range of the neighborhood of this method depends on $n$ and $\btheta_1$.
In the realistic case, we have some error for our guess of channel and the number $n$ is finite.
When the range of the error of our prior estimation of the channel
is included in the neighborhood,
this method effectively works. 
When our estimate of channel has enough precision as our prior knowledge,
we can expect such a situation.
In this scenario, we may consider the following situation.
We fix the parameter $\btheta_1$ as a basic (or standard) property of the channel.
We choose the next parameter $\btheta_2$ as a fluctuation
depending on the daily changes or the individual specificity of the channel.
Here, $n$ is chosen depending on the our calculation ability of encoding and decoding.
Hence, it is fixed priorly to the choice of $\btheta_2$.
Then, Theorem \ref{Th2} is refined as follows.

\begin{theorem}\Label{Th3}
Given real numbers $R_1^*$ and $R_2^*$,
a distribution $P$ on ${\cal X}$,
and a family of channels $\{W_{\btheta_1 +\frac{1}{\sqrt{n}}\btheta_2}\}$
with the second order parameterization,
we assume the following conditions in addition to the assumptions of Theorem \ref{Th3}.
The function $\btheta \mapsto I(P,\bW_{\btheta} )$ is a $C^1$ function on $\bTheta$,
and the function $\btheta \mapsto V(P,\bW_{\btheta} )$ is a continuous function on $\bTheta$.
Under these conditions, there exists a sequence of codes $\Phi_n$ 
with the size $|\Phi_n|=e^{nR_1^*+\sqrt{n} R_2^*- n^{\frac{1}{4}}}$ 
satisfying the following conditions.
(1) 
When $I(P,\bW_{\btheta_1} ) > R_1^*$,
we have \begin{align}
\lim_{n \to \infty}
e_{\btheta_1 +\frac{1}{\sqrt{n}}\btheta_2} (\Phi_n)=0.\Label{11-27-6C}
\end{align}
(2) When $I(P,W_{\btheta_1} )=R_1^*$,
we have 
\begin{align}
\lim_{n \to \infty}
e_{\btheta_1 +\frac{1}{\sqrt{n}}\btheta_2} (\Phi_n)
\le
\int_{-\infty}^{\frac{
R_2^*-f(\btheta_2)}{\sqrt{V(P,\bW_{\btheta} )}}}
 \frac{1}{\sqrt{2\pi}}
\exp( - \frac{x^2}{2}) dx
\Label{eq2-2C} ,
\end{align}
where
$f(\btheta_2):=
\sum_i\frac{\partial I(P,\bW_{\btheta} )}{\partial \theta^i}|_{\btheta=\btheta_1}
\theta_2^i$.
Further, 
when the convergence \eqref{12-25-12x} is uniform with respect to $\btheta_2$ in any compact subset,
there exists an upper bound 
$\bar{e}_{\btheta_1 +\frac{1}{\sqrt{n}}\btheta_2} (\Phi_n)$
of 
$e_{\btheta_1 +\frac{1}{\sqrt{n}}\btheta_2} (\Phi_n)$
such that
the bound $\bar{e}_{\btheta_1 +\frac{1}{\sqrt{n}}\btheta_2} (\Phi_n)$
converges to the RHS of \eqref{eq2-2C} uniformly for $\btheta_2$ in any compact subset.
\hfill $\square$\end{theorem}
The proof of Theorem \ref{Th3} is given in Section \ref{s8}.

Now, we consider how to realize the transmission rate close to the capacity,
In this case, we need to know that the true parameter $\btheta_0$ belongs to the neighborhood of $\btheta_1$.
More precisely, given a real number $C>0$ and the parameter $\btheta_1$, 
we need to assume that the true parameter $\btheta_0$ belongs to the set
$\{ \btheta_1+\frac{1}{\sqrt{n}}\btheta_2 | \|\btheta_2\| \le C \}$, where
$ \|\btheta\|:=\sum_{i=1}^{k} (\theta^i)^2$.
To achieve the aim, we choose the distribution $P$ such that 
$I(P,W_{\btheta_1})$ equals the capacity of the channel $W_{\btheta_1}$.
Then, we choose $R_1^*$ to be the capacity, i.e., $I(P,W_{\btheta_1})$.
To keep the average error probability approximately less than $\epsilon$, 
we choose $R_2^*$ as the maximum number satisfying
$\max_{\btheta_2: \|\btheta_2\| \le C}
\int_{-\infty}^{\frac{
R_2^*-f(\btheta_2)}{\sqrt{V(P,\bW_{\btheta} )}}}
 \frac{1}{\sqrt{2\pi}}
\exp( - \frac{x^2}{2}) dx \le \epsilon$.
Then, using Theorem \ref{Th3}, we can realize a rate close to the capacity and 
the average error probability approximately less than $\epsilon$.
When we have larger order perturbation than $\frac{1}{\sqrt{n}}$,
our upper bound of $\lim_{n \to \infty}e_{\btheta} (\Phi_n)$
becomes $0$ or $1$ because such a case corresponds to the case when
$\|\btheta_2\|$ goes to infinity.
Hence, we find that $\frac{1}{\sqrt{n}}$ is the maximum order of the perturbation
to discuss the framework of the second order.

One might doubt the validity of the above approximation, however, it can be justified as follows.
When the true parameter belongs to the neighborhood of $\btheta_1$ satisfying 
$I(P, \bW_{\btheta_1} )=R_1^*$,
the above theorem gives the approximation of average error probability.
Theorem \ref{Th3} reflects such a real situation.
Indeed, the convergence of the upper bound $\bar{e}_{\btheta} (\Phi_n)$ is not uniform for $\btheta$
in the neighborhood of such a parameter $\btheta_1$
because the limit is discontinuous at $\btheta_1$.
However, 
the convergence of $e_{\btheta_1 +\frac{1}{\sqrt{n}}\btheta_2} (\Phi_n)$
is uniform for $\btheta_2$ in any compact subset.
The reason for the importance of the second order asymptotics is that the limiting error probability can be used for the approximation of the true average error probability
because the convergence is uniform for the second order rate in any compact subset \cite{Haya}.
Due to the same reason, 
the RHS of \eqref{eq2-2C}
can be used for the approximation of our upper bound of the true average error probability.

\section{Method of types and universal encoder}\Label{s3}
In this section, we define our universal encoder for a given
distribution $P$ on ${\cal X}$ and a real positive number $ R <H(P)$
by the same way as \cite{Ha1}.
Although the contents of this section is the same as a part of \cite{Ha1},
since the paper \cite{Ha1} is written with quantum terminology,
we repeat the same contents as a part of \cite{Ha1}
for readers' convenience.

The key point of this section is to provide a code to satisfy the following property by using the method of types.
In information theory, we usually employ the random coding method.
However, to construct universal coding, we cannot employ this method because 
a code whose decoding error probability is less than the average might depend on the true channel.
Hence, we need to choose a deterministic code whose decoding error probability is less than the average.
To resolve this problem, we employ a deterministic code whose decoding error probability can be upper bounded 
by a polynomial times of the average of the decoding error probability under the random coding.
To realize this idea, we employ the method of types.

First, we prepare notations for the method of types.
Given an element ${x}^n\in {\cal X}^n$, 
we define the integer $n_x:=| \{ i| x_i=x\} |$ for $x\in {\cal X}$
and the empirical distribution 
$T_Y(x^n):= (\frac{n_1}{n},\ldots, \frac{n_d}{n})$, which is called a type.
The set of types is denoted by $T_n({\cal X})$.
For ${P} \in T_n({\cal X})$, a subset of ${\cal X}^n$ is defined by:
\begin{align*}
T_{P}:= \{{x}^n \in {\cal X}^n| \hbox{The empirical distribution of } {x}^n
\hbox{ is equal to }P\}.
\end{align*}
Hence, we define the distribution
\begin{align}
P_{T_P}(x^n):=
\left\{
\begin{array}{cc}
\frac{1}{|T_P|} & {x}^n \in T_P \\
0 & {x}^n \notin T_P.
\end{array}
\right.
\end{align}
Then, we define the constant $c_{n,P}$ by 
\begin{align}
\frac{1}{c_{n,P}} P_{T_P}({x'}^n) =  P^{n}({x'}^n)
=e^{\sum_{i=1}^d n_i \log P(i) } =e^{-nH(P)} 
\Label{eq3-2}.
\end{align}
for ${x'}^n \in T_{P}$.
So, the constant $c_{n,P}$ is bounded as
\begin{align}
c_{n,P} \le |T_n({\cal X})| 
\Label{eq3}.
\end{align}

Further, the sequence of types ${\bV}=({v}_1, \ldots, {v}_d)\in 
T_{n_1}({\cal X})\times \cdots \times T_{n_d}({\cal X})$ 
is called a conditional type for ${x}^n$
when the type of $x^n$ is $(\frac{n_1}{n}, \ldots, \frac{n_d}{n})$
\cite{CK}.
We denote the set of conditional types for ${x}^n$ 
by $V({x}^n,{\cal X})$.
For any conditional type $\bV$ for ${x}^n$, 
we define the subset of ${\cal X}^n$:
\begin{align*}
T_{{\bV}}({x}^n) 
:= \left\{{x'}^n \in {\cal X}^n\left|
T_Y((x_1,x_1'), \ldots, (x_n,x_n'))
={\bV} \cdot P
\right.\right\},
\end{align*}
where 
${P}$ is the empirical distribution of ${x}^n$.

According to Csisz\'{a}r and K\"{o}rner\cite{CK},
the proposed code is constructed as follows.
The main point of this section is to establish
that Csisz\'{a}r-K\"{o}rner's Packing lemma \cite[Lemma 10.1]{CK}
provides a code whose performance is essentially equivalent to
the average performance of random coding in the sense of (\ref{8}).

\begin{lemma}\Label{l1}
For a positive number $R>0$, there exists a sufficiently large integer $N$ satisfying the following.
For any integer $n \ge N$ and any type ${P} \in T_n({\cal X})$ satisfying $R< H(P)$,
there exist $M_n:=e^{n(R-n^{\frac{1}{4}})}$ distinct elements 
\begin{align*}
\hat{\cal M}_n:=\{E_{n,P,R}(1),\ldots, E_{n,P,R}({M_n})\} \subset T_P
\end{align*}
such that 
the inequality
\begin{align}
|T_{{\bV}}({x}^n) \cap (\hat{\cal M}_n\setminus \{{x}^n\})|
\le |T_{{\bV}}({x}^n)| 
e^{-n(H(P)-R)} 
\Label{20}
\end{align}
holds for every ${x}^n \in \hat{\cal M}_n \subset T_{{P}}$ and 
every conditional type ${\bV} \in V({x}^n,{\cal X})$.
\hfill $\square$\end{lemma}

This lemma is shown in Appendix \ref{Ap3} from Csisz\'{a}r and K\"{o}rner\cite[Lemma 10.1]{CK}.
Now we define our universal encoder $E_{n,P,R}$ by using Lemma \ref{l1}.
Note that this encoder $\hat{\cal M}_n$ does not depend on the output alphabet because the employed Packing lemma treats the conditional types from the input alphabet to the input alphabet.
Due to this property, as shown below, the decoding error probability is upper bounded by a polynomial times of the average of the decoding error probability under the random coding.
Now, we transform the property (\ref{20}) to a form applicable to our evaluation.

Using the encoder $\hat{\cal M}_n$, 
we define the distribution $P_{\hat{\cal M}_n}$ as
\begin{align*}
P_{\hat{\cal M}_n}({x}^n):= 
\left\{
\begin{array}{cc}
\frac{1}{{M}_n} & {x}^n \in \hat{\cal M}_n \\
0 & {x}^n \notin \hat{\cal M}_n.
\end{array}
\right.
\end{align*}
Now, we focus on the permutation group $S_n$ on $\{1, \ldots, n\}$.
For any ${x}^n \in {\cal X}^n$, we define an invariant subgroup 
$S_{{x}^n}\subset S_n$, where $S_n$ is the permutation group with degree $n$:
\begin{align*}
S_{{x}^n}: = \{g \in S_n | g({x}^n)={x}^n \}.
\end{align*}
For $\bV \in V({x}^n,{\cal X})$,
the probability 
$\sum_{g \in S_{{x}^n}}
\frac{1}{|S_{{x}^n}|}
P_{\hat{\cal M}_n}\circ g ({x'}^n)$
does not depend on the element ${x'}^n \in T_{{\bV}}({x}^n)  \subset T_{P}$.
Since
$
\sum_{{x'}^n \in T_{{\bV}}({x}^n)}
\sum_{g \in S_{{x}^n}}
\frac{1}{|S_{{x}^n}|}
P_{\hat{\cal M}_n}\circ g ({x'}^n)
=
\sum_{{x'}^n \in T_{{\bV}}({x}^n)}
P_{\hat{\cal M}_n} ({x'}^n)
=
|T_{{\bV}}({x}^n) \cap (\hat{\cal M}_n\setminus \{{x}^n\})|
\cdot \frac{1}{{M}_n}$,
any element ${x'}^n \in T_{{\bV}}({x}^n)  \subset T_{P}$
satisfies
$
\sum_{g \in S_{{x}^n}}
\frac{1}{|S_{{x}^n}|}
P_{\hat{\cal M}_n}\circ g ({x'}^n)
=
\frac{|T_{{\bV}}({x}^n) \cap \hat{\cal M}_n|}{|T_{{\bV}}({x}^n)|}
\cdot
\frac{1}{{M}_n}$.
Thus, any element ${x'}^n \in T_{{\bV}}({x}^n)  \subset T_{P}$
satisfies
\begin{align}
&\sum_{g \in S_{{x}^n}}
\frac{1}{|S_{{x}^n}|}
P_{\hat{\cal M}_n}\circ g ({x'}^n)
=
\frac{|T_{{\bV}}({x}^n) \cap \hat{\cal M}_n|}{|T_{{\bV}}({x}^n)|}
\cdot
\frac{1}{{M}_n} \nonumber \\
= &
\frac{|T_{{\bV}}({x}^n) \cap (\hat{\cal M}_n\setminus \{{x}^n\})|}
{|T_{{\bV}}({x}^n)|{M}_n} \nonumber \\
\stackrel{(a)}{\le} &
e^{-nH(P)}
e^{n^{\frac{1}{4}} } 
\stackrel{(b)}{=} 
{P}^{n}({x'}^n) e^{n^{\frac{1}{4}} }
=\frac{1}{c_{n,P}} P_{T_P} ({x'}^n) e^{n^{\frac{1}{4}} }
  \Label{8}
\end{align}
when the conditional type ${\bV}$ is not identical, i.e., $\bV(x|x')\neq \delta_{x,x'}$, 
where $(a)$ and $(b)$ follow from  \eqref{20}
and \eqref{eq3-2}, respectively.
Notice that $c_{n,P}$ is defined in \eqref{eq3-2}.
Relation (\ref{8}) holds for any ${x'}^n(\neq {x}^n) \in T_P$
because there exists a conditional type ${\bV}$ such that 
${x'}^n \in T_{{\bV}}({x}^n)$ and ${\bV}$ is not identical.

\section{$\alpha$-R\'enyi-relative-entropy version of Clarke-Barron formula}\Label{s4}
In this section, we discuss an $\alpha$-R\'enyi-relative-entropy version of Clarke-Barron formula \cite{CB}
for a family $\{P_{\btheta}\}_{\btheta \in \bTheta}$ of distributions on a probability space ${\cal Y}$.
In the quantum paper \cite{Ha1},
the key idea to evaluate the decoding error probability
is to upper bound the max relative entropy $D_{\max}(P_{\btheta}^n\|Q^{(n)})$ 
between each independent and identical distribution 
$P_{\btheta}^n$ and a certain distribution $Q^{(n)}$.
That is, as shown in \cite{Ha1}, when the output alphabet ${\cal Y}$ is a finite set,
we can find a distribution $Q^{(n)}$ such that
the max relative entropy $D_{\max}(P_{\btheta}^n\|Q^{(n)})$ 
behaves as $ O( \log n)$, i.e.,
\begin{align}
D_{\max}( P_{\btheta}^{n}\| Q^{(n)}) =O(\log n)
\hbox{ for } \forall \btheta \in \bTheta.
\Label{17-1} 
\end{align}
For the detail, see Remark \ref{r5-24}.
However, when the output alphabet ${\cal Y}$ is not a finite set, 
it is not easy to directly evaluate the max relative entropy.
Hence, in this paper, we focus on an $\alpha$-R\'enyi-relative-entropy version of Clarke-Barron formula as follows.
For a given distribution $\nu$ on the parameter space $\bTheta$,
we define the mixture distribution $Q_{\nu}^n[\{P_{\btheta}\}]$ with the density function 
$q_{\nu}^n[\{P_{\btheta}\}](y^n):= \int_{\bTheta} q_{\btheta}^n(y^n) \nu(\btheta)d \theta$.
When we need to clarify the family of distributions $\{P_{\btheta}\}$,
we simplify it by $Q_{\nu}^n$.



\begin{lemma}\Label{l2-A}
Assume that $\nu(\btheta)$ is continuous for $\theta$ and the support of $\nu$ is $\bTheta$.
Then, an exponential family $\{P_{\btheta}\}$ satisfies that 
\begin{align}
& D_{1+s}(P_{\btheta}^n\|Q_{\nu}^n), D_{1+n}(P_{\btheta}^n\|Q_{\nu}^n)  \le  \frac{k}{2}\log n+O(1) \Label{3-28-Ab} \\
& D_{1+s}(P_{\btheta}^n\|Q_{\nu}^n)
\le  \frac{k}{2}\log \frac{n}{2\pi}+\frac{1}{2}\log \det J_{\btheta}
+\log \frac{1}{\nu(\btheta)}- \frac{k}{2s}\log (1+s)+o(1)
\Label{3-28-A} 
\end{align}
for $s>0$, as $n$ goes to infinity.
When the continuity for $\nu(\btheta)$ is uniform for $\theta$ in any compact set,
the constant $O(1)$ on the RHS of \eqref{3-28-Ab} can be chosen uniformly in any compact set
for $\theta$.
\hfill $\square$\end{lemma}

The proof of Lemma \ref{l2-A} is given in Appendix \ref{A2}.
Clarke-Barron's paper \cite{CB} showed the relation
\begin{align}
D(P_{\btheta}^n\|Q_{\nu}^n)
= \frac{k}{2}\log n +O(1)
\Label{3-28-AT} 
\end{align}
as $n$ goes to infinity.
Since $D(P_{\btheta}^n\|Q_{\nu}^n)
\le D_{1+s}(P_{\btheta}^n\|Q_{\nu}^n) $ for $s>0$, combining \eqref{3-28-A}, we have
\begin{align}
D_{1+s}(P_{\btheta}^n\|Q_{\nu}^n)
= \frac{k}{2}\log n +O(1)
\Label{3-28-AT2} 
\end{align}
as $n$ goes to infinity.
That is, to show \eqref{3-28-AT2}, it is sufficient to show \eqref{3-28-A}.

Now, we explain how to refine the condition \eqref{17-1} with a non-finite alphabet ${\cal Y}$
by using Lemma \ref{l2-A}.
Since
\begin{align}
& P_{\btheta}^n(
\{y^n| p_{\btheta}^n(y^n) > q_{\nu}^n(y^n) e^{n\delta}\})
\nonumber \\
\le & e^{-n s' \delta+ s' D_{1+s'}( P_{\btheta}^n\| Q_{\nu}^n)},
\end{align}
(\ref{3-28-A}) implies that
\begin{align}
& -\frac{1}{n}\log 
P_{\btheta}^n (
\{y^n| p_{\btheta}^n(y^n) > q_{\nu}^n(y^n) e^{n\delta}\})
\nonumber \\
\ge &
s' \delta- \frac{s'}{n} D_{1+s'}( P_{\btheta}^n\| Q_{\nu}^n) \nonumber \\
\ge &
s' \delta- \frac{s'}{n} 
\bigg(\frac{k}{2}\log n +O(1)
\bigg) \nonumber \\
\to & s' \delta \Label{eq8B}
\end{align} 
for $s'>0$.
Since $s'$ is arbitrary,
we have
\begin{align}
\lim_{n\to \infty}
-\frac{1}{n}\log 
P_{\btheta}^n(
\{y^n| p_{\btheta}^n(y^n) >  q_{\nu}^n (y^n)e^{n\delta}\})
=
\infty \Label{eq1B}.
\end{align} 
Since $\delta>0$ can be chosen arbitrarily small,
\eqref{eq1B} can be regarded as the refinement of \eqref{17-1}. 
(See the definition of the max relative entropy in \eqref{e5-21-T}.)

Next, we consider a family of distribution $\{P_{\btheta}\}_{\btheta\in \bTheta}$
to satisfy Condition C.
In this case, we choose the discrete mixture distribution $Q_C^n$ as follows.
Since $\bTheta$ is compact, 
the subset $\bTheta_{[n]}:=  \frac{1}{\sqrt{n}} \mathbb{Z}^k \cap \bTheta$ has finite cardinality,
which increases with the order $O( n^{\frac{k}{2}})$.
Then, we define the distribution $Q_C^n[\{P_{\btheta}\}_{\btheta \in \bTheta}]$ by the probability density function
$q_C^n[\{P_{\btheta}\}_{\btheta \in \bTheta}](y^n):= \sum_{\btheta\in \bTheta_{[n]}}
\frac{1}{|\bTheta_{[n]}|} p_{\btheta}^n(y^n)$ and have the following lemma by simplifying it by $Q_C^n$.

\begin{lemma}\Label{l2-E}
When a family $\{P_{\btheta}\}$ satisfies Condition E, 
we have
\begin{align}
D_{1+s}(P_{\btheta}^n\|Q_{C}^n)
\le & \frac{k}{2}\log n+O(1)
\Label{3-28-AE} 
\end{align}
for $s>0$, as $n$ goes to infinity.
More precisely, there is an upper bound $A_{\btheta,n} $ of 
$D_{1+s}(P_{\btheta}^n\|Q_{C}^n)$ satisfying the following.
$A_{\btheta,n} $ does not depend on $s>0$,
and the quantity $A_{\btheta,n} -\frac{k}{2}\log n$ converges 
uniformly for $\btheta$.
\hfill $\square$\end{lemma}

\begin{proof}
We fix $\epsilon>0$.
Due to the assumption and the definition of $\bTheta_{[n]}$, 
we can choose an integer $N$ satisfying the following conditions.
Notice that $\max_{\btheta \in \bTheta,\btheta'\in  \bTheta_{[n]}}
\|\btheta-\btheta'\|^2=\frac{k}{4n}$.
Due to Condition C, for $n \ge N$ and $\btheta \in \bTheta$, 
we can choose $\btheta' \in \bTheta_{[n]}$ such that
$D_{1+s}(P_{\btheta}\|P_{\btheta'}) \le 
\frac{
((\max_{\xi:\|\xi \|=1}
\sum_{i,j}J_{\btheta|i,j}\xi_i \xi_j)+\epsilon) k}{8 n} $.
Since $ \frac{1}{|\bTheta_{[n]} |}p_{\btheta'}^n (y^n)
\le q_{C}^n(y^n)$, 
using the constant $B_C:=  
\max_{\btheta \in \bTheta}
\frac{
((\max_{\xi:\|\xi \|=1}
\sum_{i,j}J_{\btheta|i,j}\xi_i \xi_j)+\epsilon) k}{8 } $,
we have
\begin{align}
e^{sD_{1+s}(P_{\btheta}^n\|Q_{C}^n)}
\le 
|\bTheta_{[n]} |^s
e^{sD_{1+s}(P_{\btheta}^n\|P_{\btheta'}^n )}
\le
|\bTheta_{[n]} |^s
e^{sn \frac{B_C}{n} }
=|\bTheta_{[n]} |^s
e^{s B_C }.
\end{align}
That is,
\begin{align}
D_{1+s}(P_{\btheta}^n\|Q_{C}^n)
\le
\log |\bTheta_{[n]} | +B_C .
\end{align}
Since the compactness of $\bTheta$ implies $\log |\bTheta_{[n]} |
=\frac{k}{2}\log n +O(1)$,
we obtain the desired argument.
\end{proof}

Now, we consider Condition E, which is weaker than Condition C.
In this case, we choose the discrete mixture distribution $Q_E^n$ as follows.
Since $\bTheta^i$ is compact, 
the subset $\bTheta_{[n,i]}:=  \frac{1}{\sqrt{n}} \mathbb{Z}^k \cap \bTheta^i$ has finite cardinality,
which increases with the order $O( n^{\frac{k}{2}})$.
Then, we define the distribution $Q_E^n[\{P_{\btheta}\}, \{\bTheta^i\}]$ by the probability density function
$q_E^n[\{P_{\btheta}\}, \{\bTheta^i\}](y^n):= 
\sum_{i=1}^{\infty}\frac{6}{ \pi^2 i^2}
\sum_{\btheta\in \bTheta_{[n,i]}}
\frac{1}{|\bTheta_{[n,i]}|} p_{\btheta}^n(y^n)$ and have the following lemma by simplifying it by $Q_E^n$.

\begin{lemma}\Label{l2-F}
When a family $\{P_{\btheta}\}$ satisfies Condition E, 
we have the same inequality as \eqref{3-28-AE} with $Q_{E}^n$.
More precisely, there is an upper bound $B_{\theta,n} $ of 
$D_{1+s}(P_{\btheta}^n\|Q_{E}^n)$ satisfying the similar condition as 
$A_{\btheta,n} $.
The difference is that
the quantity $B_{\theta,n} -\frac{k}{2}\log n$ converges 
uniformly for $\theta$ in any compact subset.
because $\bTheta$ is not necessarily compact.
Hence, the obtained inequality contains the statement like \eqref{3-28-A}.
\hfill $\square$\end{lemma}

\begin{proof}
For any compact subset $\bTheta'\subset \bTheta$ and $\epsilon>0$,
we choose $i$ such that $\bTheta' \subset \bTheta_i$.
We fix $\epsilon>0$.
Due to the assumption and the definition of $\bTheta_{[n,i]}$, 
we can choose an integer $N$ satisfying the following conditions.
Notice that $\max_{\btheta \in \bTheta,\btheta'\in  \bTheta_{[n,i]}}
\|\btheta-\btheta'\|^2=\frac{k}{4n}$.
For $n \ge N$ and $\btheta \in \bTheta'$, we can choose $\btheta' \in \bTheta_{[n,i]}$ such that
$D_{1+s}(P_{\btheta}\|P_{\btheta'}) \le 
\frac{
((\max_{\xi:\|\xi \|=1}
\sum_{i,j}J_{\btheta|i,j}\xi_i \xi_j)+\epsilon) k}{8 n} $.
Since $ \frac{6}{\pi^2 i^2|\bTheta_{[n,i]} |}p_{\btheta'}^n (y^n)
\le q_{E}^n(y^n)$, 
using the constant $B_E:=  
\max_{\btheta \in \bTheta'}
\frac{
((\max_{\xi:\|\xi \|=1}
\sum_{i,j}J_{\btheta|i,j}\xi_i \xi_j)+\epsilon) k}{8 } $,
we have
\begin{align}
e^{sD_{1+s}(P_{\btheta}^n\|Q_{E}^n)}
\le 
\frac{\pi^{2s} i^{2s}}{6^s}
|\bTheta_{[n,i]} |^s
e^{sD_{1+s}(P_{\btheta}^n\|P_{\btheta'}^n )}.
\end{align}
That is,
\begin{align}
D_{1+s}(P_{\btheta}^n\|Q_{E}^n)
\le
\log \frac{\pi^{2} i^{2}}{6}
+\log |\bTheta_{[n,i]} | + B_E .
\end{align}
So, we obtain the desired argument.
\end{proof}

\begin{rem}
The paper \cite{HO} seems discuss a topic related to $\alpha$-R\'enyi-relative-entropy version of Clarke-Barron formula.
The main issue of the paper \cite{HO} is to evaluate $D(P_{\btheta}^n\|Q_{\nu}^n) $
not but $D_{\alpha}(P_{\btheta}^n\|Q_{\nu}^n) $.
For this purpose, the paper \cite{HO} discusses the $\alpha$-R\'enyi relative entropy 
$D_{\alpha}(P_{\btheta}\|P_{\btheta'}) $ between two close points $\theta$ and $\btheta'$.
\end{rem}

\begin{rem}\Label{r5-24}
In the relation to \eqref{eq1B}, we consider the finite output alphabet case under the same assumption as Lemma \ref{l2-A}.
We additionally assume that $\{P_{\btheta}\}$ equals the set of distributions on ${\cal Y}$.
So, for any type $P \in T_n({\cal Y})$, there exists a parameter $\btheta_P \in \bTheta $
such that $ P=P_{\btheta_P}$.
For a constant $c>0$, a type $P \in T_n({\cal Y})$, and an positive integer $n$,
there exists a constant $c_{n,P,c}$ such that
when the distribution $P_{\btheta}$ is close to the type $P \in T_n({\cal Y})$, i.e.,
$D(P_{\btheta}\| P)\ge \frac{c}{n}$, we have
$c_{n,P,c} P_{\btheta}^n (y^n) \ge  P_{T_P}(y^n) $.
Similar to \eqref{eq3}, the constant $c_{n,P,c}$ increases only polynomially with respect to $n$.

Since $D(P_{\btheta}\| P_{\btheta_P})\cong 
\frac{1}{2}\sum_{i,j} J_{\btheta|i,j}(\theta^i-\theta_P^i) (\theta^j-\theta_P^j)$, the probability 
$\nu (B_{P,n,c})$ of the set $B_{P,n,c}:=\{ \btheta \in \bTheta | D(P_{\btheta}\| P)\ge \frac{c}{n} \}$
behaves as the inverse of a polynomial of $n$.
The mixture distribution $ Q_{\nu}$ is evaluated as
$ Q_{\nu}^n(y^n) \ge \sum_{P \in T_n({\cal X})} 
\nu (B_{P,n,c}) c_{n,P,c}^{-1}P^n(y^n)$.
This fact shows that
for any $\delta>0$, there exists $N$ such that we have
\begin{align}
  e^{n\delta} 
Q_{\nu}^n(y^n) \ge \sum_{P \in T_n({\cal Y})}  P(y^n)
\ge P_{\btheta}^n(y^n)
\Label{Eq224}
\end{align}
for any integer $n \ge N$ and any $\btheta \in \bTheta$.
Hence, $ D_{\max}(P_{\btheta}^n \| Q_{\nu}^n) \le n \delta$,
which is a stronger statement than \eqref{eq1B}.

More strongly, the quantum paper \cite{Ha1} 
showed that there exists a distribution $Q^{(n)}$ such that
$ D_{\max}(P_{\btheta}^n \| Q^{(n)}) $ behaves as $O(\log n)$.
That is, the paper \cite{Ha1} essentially chosen such a distribution $Q^{(n)}$ in \eqref{17-1} to be the uniform mixture of the uniform distributions $P_{T_P}$, i.e., 
$\sum_{P \in T_n({\cal Y})} \frac{1}{|T_n({\cal Y})|}P_{T_P}$\cite{TH}.  
This distribution $ Q^{(n)}$ satisfies $ D_{\max}(P_{\btheta}^n \| Q^{(n)})
\le \log |T_n({\cal Y})|$.
The paper \cite{Ha1} employed this property to construct the universal channel coding even in the quantum setting.
However, in the infinite output alphabet case, 
we cannot show \eqref{Eq224}. 
Hence, we need this long discussion here.
The inequality \eqref{Eq224} in finite systems
also played an essential role in the papers \cite{TH,HT,Ha2}.
\end{rem}

\section{Universal decoder}\Label{s5}
The aim of this section is to make a universal decoder, 
which does not depend on the parameter $\theta$
and depends only on 
the distribution $P$, the family of channel $\{\bW_{\btheta}\}$,
prior distributions on $\bTheta$,
the coding rate $R$, and another rate $R_1$.
The other rate $R_1$ is decided by the distribution $P$, the family $\{\bW_{\btheta}\}$, and $R$
as discussed in the method (M1) or (M2) (large deviation) 
and Section \ref{s8} (second order).
The required property holds for any choice of the prior distributions on $\bTheta$.

\subsection{Exponential family of channels}
In this subsection, we construct a decoder only for 
exponential family of channels, i.e., under Condition B.
In Subsection \ref{s5-2}, under Condition D or F, we can construct it in the same way by replacing
$Q_w^{n}$ by $Q_C^{n}$ or $Q_E^{n}$ due to Lemmas \ref{l2-E} and \ref{l2-F}.
In the decoding process, we need to extract the message that has large correlation
with detected output signal like the maximum mutual information decoder \cite{CK}.
However, since the cardinality of the output alphabet ${\cal Y}$ is not necessarily finite, we cannot directly apply the maximum mutual information decoder
because it can be defined only with finite input and output alphabets.
Here, to overcome this problem,
we need to extract the message that has large correlation with detected output signal.
To achieve this aim, we firstly construct 
output distributions that universally approximate
the true output distribution and the mixture of the output distribution.
Then, we apply the decoder for information spectrum method \cite{Verdu-Han} to these distributions.
Using this idea, we construct our universal decoder.
That is, for a given number $R>0$ and a distribution $P$ on ${\cal X}$,
we define our universal decoder for our universal encoder $E_{n,P,R}$
given in Section \ref{s3} by using the type $P=
(\frac{n_1}{n}, \ldots, \frac{n_d}{n})$.

We choose prior distributions $\nu$ with support $\bTheta$.
For a given $x^n\in T_P$, we define
the distribution $Q_{x^n}^{(n)}$ on ${\cal Y}^n$ as follows.
First, for simplicity, we consider the case 
when 
\begin{align}
{x'}^n=(\underbrace{1,\ldots,1}_{n_1},\underbrace{2,\ldots,2}_{n_2},\ldots, 
\underbrace{d,\ldots, d}_{n_d}) \in T_P.
\Label{E5-21-7}
\end{align}
In this case, 
$Q_{{x'}^n}^{(n)}$ is defined as
$Q_{\nu}^{n_1}[\{\bW_{\btheta,1}\}]\times Q_{\nu}^{n_2} [\{\bW_{\btheta,2}\}]\times \cdots \times Q_{\nu}^{n_d}[\{\bW_{\btheta,d}\}]$. 
For a general element $x^n \in T_P$, 
the distribution $Q_{{x}^n}^{(n)}$ on ${\cal Y}^n$
is defined by 
the permuted distribution of $Q_{{x'}^n}^{(n)}$ by 
the application of the permutation 
converting ${x'}^n$ to $x^n$.

Now, to make our decoder,
we apply the decoder for information spectrum method \cite{Verdu-Han} to  
the case when 
the output distribution is given by $Q_{{x}^n}^{(n)}$
and the mixture of the output distribution is given by 
$Q_{P}^{(n)} := \sum_{x^n \in T_P} P_{T_P}({x}^n) Q_{{x}^n}^{(n)}$.
Then, we define the subset 
\begin{align}
\hat{{\cal D}}_{x^n}:= \{y^n| 
q_{{x}^n}^{(n)} (y^n) \ge e^{n R_1} q_{P}^{(n)} (y^n)
\}.
\Label{11-27-6}
\end{align}
Given a number $R_1 >R$ and an element $i \in {\cal M}_n$, 
we define the subset 
${\cal D}_{i}:= 
\hat{{\cal D}}_{E_{n,P,R}(i)} \setminus 
\cup_{j=1}^{i-1}\hat{{\cal D}}_{E_{n,P,R}(i)}$,
inductively.
Finally, we define our universal decoder 
$D_{n,P,R,R_1}$ as $D_{n,P,R,R_1}^{-1}(i):= {\cal D}_{i}$
and our universal code 
$\Phi_{n,P,R,R_1}:=(e^{nR- n^{\frac{1}{4}}},E_{n,P,R},D_{n,P,R,R_1})$.

Now, we discuss several important properties related to our universal decoder
by using Lemma \ref{l2-A}.

For $x^n \in T_P$ and $s>0$,
\eqref{3-28-A} of Lemma \ref{l2-A} guarantees that
\begin{align}
& D_{1+s}( W^{n}_{\btheta,x^n}\| Q_{{x}^n}^{(n)}) \nonumber \\
\le & \sum_{x \in {\cal X}}
\Bigl(\frac{k}{2}\log (n P(x)) \Bigr)+O(1).
\Label{3-28-1}
\end{align}
Hence, similar to \eqref{eq8B}, (\ref{3-28-1}) implies that
\begin{align}
\liminf_{n\to \infty}
-\frac{1}{n}\log 
W^{n}_{\btheta,x^n}(
\{y^n| w^{n}_{\btheta,x^n}(y^n) > e^{n\delta} q_{{x}^n}^{(n)}(y^n) \})
\ge  s' \delta \Label{eq8}
\end{align} 
for $s'>0$.
Since $s'$ is arbitrary,
we have
\begin{align}
\lim_{n\to \infty}
-\frac{1}{n}\log 
W^{n}_{\btheta,x^n}(
\{y^n| w^{n}_{\btheta,x^n} (y^n) > e^{n\delta} q_{{x}^n}^{(n)}(y^n) \})
=
\infty \Label{eq1}.
\end{align} 

Modifying the derivation in \eqref{eq8B},
we have
\begin{align}
-\frac{1}{\sqrt{n}}\log 
W^{n}_{\btheta,x^n}(
\{y^n| w^{n}_{\btheta,x^n}(y^n) > e^{\sqrt{n}\delta} q_{{x}^n}^{(n)}(y^n) \})
\ge  s' \delta +o(1)\Label{eq8T}
\end{align} 
for $s'>0$.
Here, the term $o(1)$ in \eqref{eq8T} can be chosen as a arbitrary small constant uniformly with respect to $n$ and ${\btheta}$ in any compact set due to Lemma \ref{l2-A}.

Further, since
$Q_{P}^{(n)} = \sum_{x^n \in T_P} P_{T_P}(x^n) Q_{{x}^n}^{(n)}$,
the information processing inequality for R\'{e}nyi relative entropy yields that
$e^{s D_{1+s}( \bW^{n}_{\btheta}\cdot P_{T_P}\| Q_{P}^{(n)})}
\le
\sum_{x^n \in T_P} P_{T_P} (x^n)e^{s D_{1+s}( W^{n}_{\btheta,x^n}\| Q_{{x}^n}^{(n)})}
=
e^{s D_{1+s}( W^{n}_{\btheta,x^n}\| Q_{{x}^n}^{(n)})}$, which implies 
\begin{align}
& D_{1+s}( W^{n}_{\btheta}\cdot P_{T_P}\| Q_{P}^{(n)}) \nonumber \\
\le & \sum_{x \in {\cal X}}
\Bigl(\frac{k}{2}\log (n P(x))\Bigr)+O(1).
\Label{3-28-1b}
\end{align}
Here, the constant $O(1)$ can be chosen uniformly with respect to ${\btheta}$ in any compact set 
 due to Lemma \ref{l2-A}.

Similar to \eqref{eq1}, \eqref{3-28-1b} implies that
\begin{align}
\lim_{n \to \infty} \frac{-1}{n}\log 
\bW_{\btheta}^n\cdot P_{T_P}
\{y^n| 
\bw_{\btheta}^n\cdot P_{T_P}(y^n) > e^{n \delta} q_P^{(n)}  (y^n)
\} 
=\infty.\Label{eq13UU}
\end{align}

\subsection{General assumption}\Label{s5-2}
Next, under more general assumption, i.e., Conditions D and F, we construct 
the universal decoder by replacing
$Q_w^{n}$ by $Q_C^{n}$ or $Q_E^{n}$ due to Lemmas \ref{l2-E} and \ref{l2-F}.
Under Condition D, first, 
for simplicity, we consider the case of \eqref{E5-21-7}.
In this case, 
$Q_{{x'}^n}^{(n)}(y^n)$ is defined as
$Q_{C}^{n_1}[\{\bW_{\btheta,1}\}_{\btheta \in \bTheta}]\times 
Q_{C}^{n_2} [\{\bW_{\btheta,2}\}_{\btheta \in \bTheta}]\times \cdots \times 
Q_{C}^{n_d}[\{\bW_{\btheta,d}\}_{\btheta \in \bTheta}]$. 
For a general element $x^n \in T_P$, 
the distribution $Q_{{x}^n}^{(n)}$ on ${\cal Y}^n$
is defined by 
the permuted distribution of $Q_{{x'}^n}^{(n)}$ by 
the application of the permutation 
converting ${x'}^n$ to $x^n$.
Using $Q_{P}^{(n)} := \sum_{x^n \in T_P} P_{T_P}(x^n) Q_{{x}^n}^{(n)}$,
we define the decoder as \eqref{11-27-6}.
So, in the same way, we can show \eqref{eq1}, \eqref{eq8T}, \eqref{3-28-1b}, and \eqref{eq13UU} by using \eqref{3-28-AE} of Lemma \ref{l2-E}.
In particular, the constants $o(1)$ and $O(1)$ in \eqref{eq8T} and \eqref{3-28-1b} can be chosen uniformly 
with respect to ${\btheta}$ due to the same reason, respectively.

Under Condition F, in the case of \eqref{E5-21-7},
the distribution $Q_{{x'}^n}^{(n)}$ is defined as
$Q_{E}^{n_1}[\{\bW_{\btheta,1}\}, \{ \bTheta^i\}]\times 
Q_{E}^{n_2} [\{\bW_{\btheta,2}\}, \{ \bTheta^i\}]\times \cdots \times 
Q_{E}^{n_d}[\{\bW_{\btheta,d}\}, \{ \bTheta^i\}]$. 
For a general element $x^n \in T_P$, 
the distribution $Q_{{x}^n}^{(n)}$ on ${\cal Y}^n$
is defined by 
the permuted distribution of $Q_{{x'}^n}^{(n)}$ by 
the application of the permutation 
converting ${x'}^n$ to $x^n$.
Using the distribution $Q_{P}^{(n)} := \sum_{x^n \in T_P} P_{T_P}(x^n) Q_{{x}^n}^{(n)}$,
we define the decoder as \eqref{11-27-6}.
So, in the same way, we can show \eqref{eq1}, \eqref{eq8T}, \eqref{3-28-1b}, and \eqref{eq13UU} by using Lemma \ref{l2-E}.
Again, the constants $o(1)$ and $O(1)$ in \eqref{eq8T} and \eqref{3-28-1b} can be chosen uniformly 
with respect to ${\btheta}$ due to the same reason, respectively.

\section{Error exponent}\Label{s6}
In this section, 
using the property \eqref{eq3} of type, the encoder property \eqref{8},
and the decoder properties \eqref{eq1} and \eqref{eq13UU},   
we will prove Theorem \ref{Th1}, i,e., show that the code
$\Phi_{n,P,R,R_1}$ satisfies that
\begin{align}
\lim_{n \to \infty}\frac{-1}{n} 
\log e_{\btheta} (\Phi_{n,P,R,R_1})
\ge \min( \max_{s \in [0,1]}
 (s I_{1-s}(P,\bW_{\btheta} ) -s R_1), R_1-R ). \Label{eq5}
\end{align}
Since  the decoder properties \eqref{eq1} and \eqref{eq13UU} holds for all of Conditions B, D, and F,
the following proof is valid under Conditions B, D, and F.

First, we have
\begin{align}
& \sum_{i=1}^{M_n}\frac{1}{M_n}
W^{n}_{{\btheta},E_{n,P,R}(i)}({\cal D}_{i}^c) \nonumber \\
\le &
\sum_{i=1}^{M_n}\frac{1}{M_n} \Big(
(W^{n}_{{\btheta},E_{n,P,R}(i)}(\hat{{\cal D}}_{E_{n,P,R}(i)}^c) \nonumber \\
&\qquad \qquad+
\sum_{j\neq i}W^{n}_{{\btheta},E_{n,P,R}(i)}(\hat{{\cal D}}_{E_{n,P,R}(j)}))\Big)
\nonumber \\
= &
\sum_{i=1}^{M_n}\frac{1}{M_n}
\Big( W^{n}_{{\btheta},E_{n,P,R}(i)}(\hat{{\cal D}}_{E_{n,P,R}(i)}^c) \nonumber \\
&+
\sum_{j=1}^{M_n} \frac{1}{M_n}
\sum_{i \neq j} W^{n}_{{\btheta},E_{n,P,R}(i)}(\hat{{\cal D}}_{E_{n,P,R}(j)})\Big) .
\Label{eq6}
\end{align}

In the following, we evaluate the first term of (\ref{eq6}).
Any element $x^n \in T_P$ satisfies that
\begin{align}
& W^{n}_{{\btheta},x^n}(\hat{{\cal D}}_{x^n}^c) \nonumber \\
=&
W^{n}_{{\btheta},x^n}(
\{y^n| q_{{x}^n}^{(n)} (y^n) < e^{n R_1} q_{P}^{(n)} (y^n) \}) \nonumber \\
\le &
W^{n}_{{\btheta},x^n}(
\{y^n| w^{n}_{{\btheta},x^n}(y^n) < e^{n (R_1+\delta)} q_{P}^{(n)} (y^n) \})
\nonumber \\
&+
W^{n}_{{\btheta},x^n}(
\{y^n| w^{n}_{{\btheta},x^n}(y^n) > e^{n\delta} q_{{x}^n}^{(n)}(y^n) \})
\Label{eq7}
\end{align}
because when 
$w^{n}_{{\btheta},x^n}(y^n) \le e^{n\delta} q_{{x}^n}^{(n)}(y^n)$,
the condition $q_{{x}^n}^{(n)} (y^n) < e^{n R_1} q_{P}^{(n)} (y^n)$
implies $w^{n}_{{\btheta},x^n}(y^n) < e^{n (R_1+\delta)} q_{P}^{(n)} (y^n)$.
To get the latter exponent, we choose any elements 
$s\in [0,1]$ and $x^n \in T_P$. 
Using the measure
$\mu^n(dy^n):=\mu(dy_1)\cdots \mu(dy_n)$,
we have
\begin{align}
& W^{n}_{{\btheta},x^n}(
\{y^n| w^{n}_{{\btheta},x^n}(y^n) < e^{n (R_1+\delta)} q_{P}^{(n)} (y^n) \})\nonumber \\
\le &
\int_{{\cal Y}^n}
w^{n}_{{\btheta},x^n}(y^n)^{1-s}
 e^{n s(R_1+\delta)} q_{P}^{(n)} (y^n)^s \mu^n(dy^n)
\nonumber \\
\stackrel{(a)}{=} &
\sum_{x^n \in T_P}\frac{1}{|T_P|}
\int_{{\cal Y}^n} 
 w^{n}_{{\btheta},x^n}(y^n)^{1-s}
 e^{n s(R_1+\delta)} q_P^{(n)} (y^n)^s 
\mu^n(dy^n)
\nonumber \\
\stackrel{(b)}{\le} &
|T_n({\cal X})|
\nonumber \\
& \cdot
\sum_{x^n \in {\cal X}^n}P^n(x^n)
\int_{{\cal Y}^n} 
 w^{n}_{{\btheta},x^n}(y^n)^{1-s}
 e^{n s(R_1+\delta)} q_P^{(n)} (y^n)^s \mu^n(dy^n)
\nonumber \\
\stackrel{(c)}{=} &
|T_n({\cal X})|
e^{n s(R_1+\delta)}
\nonumber \\
& \cdot
\int_{{\cal Y}^n} 
\Big(\sum_{x^n \in {\cal X}^n}P^n(x^n)
 w^{n}_{{\btheta},x^n}(y^n)^{1-s}\Big)
  q_P^{(n)} (y^n)^s \mu^n(dy^n)
\nonumber \\
\stackrel{(d)}{\le} &
|T_n({\cal X})|
e^{n s(R_1+\delta)}
\nonumber \\
& \cdot
\Big(\int_{{\cal Y}^n} 
\Big(\sum_{x^n \in {\cal X}^n}P^n(x^n)
 w^{n}_{{\btheta},x^n}(y^n)^{1-s}\Big)^{\frac{1}{1-s}}
\mu^n(dy^n) \Big)^{1-s} \nonumber \\
& \cdot 
\Big(\int_{{\cal Y}^n} q_P^{(n)} (y^n)^\frac{s}{s} \mu^n(dy^n)\Big)^{\frac{1}{s}}
\nonumber \\
=&
|T_n({\cal X})|
e^{n s(R_1+\delta)}
\nonumber \\
& \cdot
\Big(\int_{{\cal Y}^n} 
\Big(\sum_{x^n \in {\cal X}^n}P^n(x^n)
 w^{n}_{{\btheta},x^n}(y^n)^{1-s}\Big)^{\frac{1}{1-s}}
\mu^n(dy^n) \Big)^{1-s} 
\nonumber \\
= &
|T_n({\cal X})|
e^{n s(R_1+\delta)}
\nonumber \\
& \cdot
\Big(\int_{{\cal Y}} 
\Big(\sum_{x \in {\cal X}}P (x)
 w_{{\btheta},x}(y)^{1-s}\Big)^{\frac{1}{1-s}}
\mu(dy)\Big)^{n(1-s)}  \nonumber \\
=& |T_n({\cal X})|
e^{n (s(R_1+\delta)-s I_{1-s}(P,\bW_{\btheta} ) )},\Label{eq9}
\end{align}
where $(b)$ and $(d)$ follow from 
(\ref{eq3}) and the H\"{o}lder inequality, respectively.
Step $(c)$ follows from the exchange of the orders of finite sum and the integral.

Since the value of the integral 
$\int_{{\cal Y}^n} w^{n}_{{\btheta},x^n}(y^n)^{1-s}
 e^{n s(R_1+\delta)} q_{P}^{(n)} (y^n)^s \mu^n(dy^n) $
does not change when the order of $y_1, \ldots, y_n$,
we have the relation $$
\int_{{\cal Y}^n} w^{n}_{{\btheta},x^n}(y^n)^{1-s}
 e^{n s(R_1+\delta)} q_{P}^{(n)} (y^n)^s \mu^n(dy^n)=
\int_{{\cal Y}^n} w^{n}_{{\btheta},{x^n}'}(y^n)^{1-s}
 e^{n s(R_1+\delta)} q_{P}^{(n)} (y^n)^s \mu^n(dy^n)$$
 with ${x^n}'\neq {x^n}\in T_P $, which implies Step $(a)$.

Due to \eqref{eq1} and \eqref{eq7}, 
the exponential decreasing rate of $W^{n}_{{\btheta},x^n}(\hat{{\cal D}}_{x^n}^c)$
equals that of 
$w^{n}_{{\btheta},x^n}(\{y^n| w^{n}_{{\btheta},x^n}(y^n) < e^{n (R_1+\delta)} q_{P}^{(n)} (y^n) \})$.
Thus, using \eqref{eq9}, we have
\begin{align*}
\lim_{n\to \infty}
-\frac{1}{n}\log W^{n}_{{\btheta},x^n}(\hat{{\cal D}}_{x^n}^c) 
\ge 
-s(R_1+\delta) +s I_{1-s}(P,\bW_{\btheta} ).
\end{align*} 
Since $\delta>0$ is arbitrary, we have
\begin{align}
\lim_{n\to \infty}
-\frac{1}{n}\log W^{n}_{{\btheta},x^n}(\hat{{\cal D}}_{x^n}^c) 
\ge 
-s R_1+ s I_{1-s}(P,\bW_{\btheta} ).
\Label{eq10}
\end{align} 

Next, we proceed to the second term of (\ref{eq6}).
In the following, we simplify 
$\bW_{{\btheta}}\cdot P $ to be $W_{{\btheta},P} $. 
\begin{align}
& \frac{1}{M_n} \sum_{j=1}^{M_n} \sum_{i \neq j}
W^{n}_{{\btheta},E_{n,P,R}(i)}(\hat{{\cal D}}_{E_{n,P,R}(j)}) \nonumber\\
\stackrel{(a)}{=}  &
\sum_{j=1}^{M_n} \sum_{i\neq j}
\sum_{g\in S_{E_{n,P,R}(j)}}\frac{1}{|S_{E_{n,P,R}(j)}|}
\frac{1}{M_n}
W^{n}_{{\btheta},g (E_{n,P,R}(i))}(\hat{{\cal D}}_{E_{n,P,R}(j)}) \nonumber\\
= &
\sum_{j=1}^{M_n}
\sum_{g\in S_{E_{n,P,R}(j)}}\frac{1}{|S_{E_{n,P,R}(j)}|}
\sum_{{x'}^n (\neq E_{n,P,R}(j) )\in T_P}
P_{\hat{\cal M}_{n}} (g({x'}^n))
W^{n}_{{\btheta},{x'}^n}(\hat{{\cal D}}_{E_{n,P,R}(j)})
\nonumber\\
\stackrel{(b)}{\le}  &
\frac{e^{n^{\frac{1}{4}}}}{c_{n,P}}
\sum_{j=1}^{M_n}
\sum_{x^n \in {\cal X}^n}
P_{T_P}(x^n)
W^{n}_{{\btheta},x^n}(\hat{{\cal D}}_{E_{n,P,R}(j)}) \nonumber\\
=&
\frac{e^{n^{\frac{1}{4}}}}{c_{n,P}}
\sum_{j=1}^{M_n}
\bW_{{\btheta}}^n\cdot P_{T_P}(\hat{{\cal D}}_{E_{n,P,R}(j)}) \nonumber\\
=&
e^{n^{\frac{1}{4}}}
\sum_{j=1}^{M_n}
W_{{\btheta}}^n\cdot P_{T_P}
\{y^n| 
q_{E_{n,P,R}(j)}^{(n)} (y^n) \ge e^{n R_1} q_P^{(n)} (y^n)
\} \nonumber\\
\le &
\frac{e^{n^{\frac{1}{4}}}}{c_{n,P}}
\sum_{j=1}^{M_n}
\bW_{{\btheta}}^n\cdot P_{T_P}
\{y^n| 
q_{E_{n,P,R}(j)}^{(n)} (y^n) \ge e^{n (R_1-\delta)} w_{{\btheta}}^n\cdot P_{T_P} (y^n)
\} \nonumber\\
& +
\frac{e^{n^{\frac{1}{4}}}}{c_{n,P}}
\sum_{j=1}^{M_n}
\bW_{{\btheta}}^n\cdot P_{T_P}
\{y^n| 
\bw_{{\btheta}}^n\cdot P_{T_P}(y^n) > e^{n \delta} q_P^{(n)}(y^n)
\} \nonumber\\
\le &
\frac{e^{n^{\frac{1}{4}}}}{c_{n,P}}
\sum_{j=1}^{M_n}
e^{-n (R_1-\delta)}
Q_{E_{n,P,R}(j)}^{(n)}
\{y^n| 
q_{E_{n,P,R}(j)}^{(n)} (y^n) \ge e^{n (R_1-\delta)} w_{{\btheta}}^n\cdot P_{T_P} (y^n)
\}
\nonumber \\
& +
\frac{e^{n^{\frac{1}{4}}}}{c_{n,P}}
\sum_{j=1}^{M_n}
\bW_{{\btheta}}^n\cdot P_{T_P}
\{y^n| 
 \bw_{{\btheta}}^n\cdot P_{T_P} (y^n) > e^{n \delta} q_P^{(n)} (y^n)
\} \nonumber \\
\le  &
\frac{e^{n^{\frac{1}{4}}}}{c_{n,P}}
\sum_{j=1}^{M_n}
e^{-n (R_1-\delta)}
 +
\frac{e^{n^{\frac{1}{4}}}}{c_{n,P}}
M_n
\bW_{{\btheta}}^n\cdot P_{T_P}
\{y^n| 
\bw_{{\btheta}}^n\cdot P_{T_P} (y^n) > e^{n \delta} q_P^{(n)} (y^n)
\} \nonumber \\
=  &
\frac{e^{n^{\frac{1}{4}}}}{c_{n,P}}
M_n
e^{-n (R_1-\delta)} +
e^{\sqrt{n}}
M_n
\bW_{{\btheta}}^n\cdot P_{T_P}
\{y^n| 
\bw_{{\btheta}}^n\cdot P_{T_P} (y^n) > e^{n \delta} q_P^{(n)} (y^n)
\} 
\Label{eq11}
\end{align}
where 
$(a)$ follows from the relation $
W^{n}_{{\btheta},g (E_{n,P,R}(i))}(\hat{{\cal D}}_{E_{n,P,R_1}(j)})
=W^{n}_{{\btheta},E_{n,P,R}(i)}(\hat{{\cal D}}_{E_{n,P,R_1}(j)}) $ 
for $g\in S_{E_{n,P,R}(j)}$
and $(b)$ follows from (\ref{8}).

Since 
$\bW_{{\btheta}}^n\cdot P_{T_P}
\{y^n| \bw_{{\btheta}}^n\cdot P_{T_P} (y^n)
 > e^{n \delta} q_P^{(n)} (y^n)
 \} $
satisfies the condition (\ref{eq13UU})
and\par
\noindent$\lim_{n \to \infty} \frac{1}{n}\log 
\frac{e^{n^{\frac{1}{4}}}}{c_{n,P}}
M_n= \lim_{n \to \infty} \frac{1}{n}\log 
M_n < \infty$, 
\begin{align}
\lim_{n \to \infty} \frac{-1}{n}\log 
\frac{e^{n^{\frac{1}{4}}}}{c_{n,P}}
M_n
\bW_{{\btheta}}^n\cdot P_{T_P}
\{y^n| 
\bw_{{\btheta}}^n\cdot P_{T_P} (y^n) > e^{n \delta} q_P^{(n)} (y^n)
\} 
=\infty.\Label{eq13}
\end{align}
Since
$ \lim_{n \to \infty} \frac{-1}{n}\log 
\frac{e^{n^{\frac{1}{4}}}}{c_{n,P}}
M_n e^{-n (R_1-\delta)} 
=R_1-\delta-R$,
\eqref{eq11} and \eqref{eq13} imply that
\begin{align}
\lim_{n \to \infty} \frac{-1}{n}\log 
 \frac{1}{M_n} \sum_{j=1}^{M_n} \sum_{i \neq j}
W^{n}_{{\btheta},E_{n,P,R}(i)}(\hat{{\cal D}}_{E_{n,P,R}(j)}) 
\ge R_1-\delta-R.
\end{align}
Since $\delta>0$ is arbitrary, we have
\begin{align}
\lim_{n \to \infty} \frac{-1}{n}\log 
 \frac{1}{M_n} \sum_{j=1}^{M_n} \sum_{i \neq j}
W^{n}_{{\btheta},E_{n,P,R}(i)}(\hat{{\cal D}}_{E_{n,P,R}(j)}) 
\ge R_1-R.\Label{11-27-3}
\end{align}
Combining (\ref{eq6}), (\ref{eq10}), and (\ref{11-27-3}), we obtain (\ref{eq5}).


\section{Second order}\Label{s8}
In this section, we show Theorems \ref{Th2} and \ref{Th3}.
First, we show Theorem \ref{Th2}. Next, we show Theorem \ref{Th3} by modifying the proof of Theorem \ref{Th2}.
These proofs are based on 
the decoder properties \eqref{eq8T} and \eqref{3-28-1b}
as well as the discussion in Section \ref{s6}.
To satisfy the condition of Theorem \ref{Th2}, we modify 
the encoder given in Section \ref{s3} 
and the decoder given in Section \ref{s5} as follows.
For our encoder, we choose $R$ to be $R_1^*+\frac{R_2^*}{\sqrt{n}}$,
and for our decoder, we choose $R_1$ to be $R_1^*+\frac{R_2^*}{\sqrt{n}}+ \frac{1}{n^{2/3}}$.
The modified code is denoted by $\Phi_{n,P,R_1^*,R_2^*}$.

\subsection{Evaluation of second term in \eqref{eq6}}
We evaluate the first and the second terms in \eqref{eq6}, separately.
Here, we show that the second term in \eqref{eq6} goes to zero.
Similar to \eqref{eq11}, 
we have
\begin{align}
& \frac{1}{M_n} \sum_{j=1}^{M_n} \sum_{i \neq j}
W^{n}_{{\btheta},E_{n,P,R}(i)}(\hat{{\cal D}}_{E_{n,P,R}(j)}) \nonumber\\
=& 
\frac{e^{n^{\frac{1}{4}}}}{c_{n,P}}
\sum_{j=1}^{M_n}
\bW_{{\btheta}}^n\cdot P_{T_P}(\hat{{\cal D}}_{E_{n,P,R}(j)}) \nonumber\\
=& 
\frac{e^{n^{\frac{1}{4}}}}{c_{n,P}}
\sum_{j=1}^{M_n}
\int_{\hat{{\cal D}}_{E_{n,P,R}(j)}}
\bw_{{\btheta}}^n\cdot P_{T_P}(y^n) q_P^{(n)} (y^n)^{s-1} Q_P^{(n)} (y^n)^{1-s}
\mu^n(dy^n)
\nonumber\\
\stackrel{(a)}{\le} &
\frac{e^{n^{\frac{1}{4}}}}{c_{n,P}}
\sum_{j=1}^{M_n}
\Big(\int_{\hat{{\cal D}}_{E_{n,P,R}(j)}}
\bw_{{\btheta}}^n\cdot P_{T_P}(y^n)^{\frac{1}{s}} q_P^{(n)} (y^n)^{\frac{s-1}{s}} 
\mu^n(dy^n)\Big)^{s}
\Big(\int_{\hat{{\cal D}}_{E_{n,P,R}(j)}}
q_P^{(n)} (y^n)^{\frac{1-s}{1-s}}
\mu^n(dy^n)\Big)^{1-s}
\nonumber\\
\le &
\frac{e^{n^{\frac{1}{4}}}}{c_{n,P}}
\sum_{j=1}^{M_n}
\Big(\int_{{\cal Y}^n}
\bw_{{\btheta}}^n\cdot P_{T_P}(y^n)^{\frac{1}{s}} Q_P^{(n)} (y^n)^{\frac{s-1}{s}} 
\mu^n(dy^n)\Big)^{s}
Q_P^{(n)} ( \hat{{\cal D}}_{E_{n,P,R}(j)})^{1-s}
\nonumber\\
= &
\frac{e^{n^{\frac{1}{4}}}}{c_{n,P}}
\sum_{j=1}^{M_n}
e^{s(\frac{1}{s}-1)D_{\frac{1}{s}} (\bW_{{\btheta}}^n\cdot P_{T_P}
\|Q_P^{(n)} ) }
Q_P^{(n)} ( \hat{{\cal D}}_{E_{n,P,R}(j)})^{1-s}
\nonumber\\
= &
\frac{e^{n^{\frac{1}{4}}}}{c_{n,P}}
\sum_{j=1}^{M_n}
e^{(1-s)D_{\frac{1}{s}} (\bW_{{\btheta}}^n\cdot P_{T_P}
\|Q_P^{(n)} ) }
Q_P^{(n)} ( \hat{{\cal D}}_{E_{n,P,R}(j)})^{1-s} 
\nonumber\\
=&
\frac{e^{n^{\frac{1}{4}}}}{c_{n,P}}
M_n
e^{(1-s)D_{\frac{1}{s}} (\bW_{{\btheta}}^n\cdot P_{T_P}
\|Q_P^{(n)} ) }
Q_P^{(n)} ( \hat{{\cal D}}_{E_{n,P,R}(j)})^{1-s} 
\nonumber\\
\le &
\frac{e^{n^{\frac{1}{4}}}}{c_{n,P}}
e^{n R_1^* + \sqrt{n}R_2^*-n^{\frac{1}{4}} }
e^{(1-s)D_{\frac{1}{s}} (\bW_{{\btheta}}^n\cdot P_{T_P} \|Q_P^{(n)} ) }
e^{-n R_1^* - \sqrt{n}R_2^* - n^{\frac{1}{3}} } 
Q_{ E_{n,P,R}(j) }^{(n)} ( \hat{{\cal D}}_{E_{n,P,R}(j)})^{1-s} 
\nonumber\\
= &
\frac{1}{c_{n,P}}
e^{- n^{\frac{1}{3}} } 
e^{(1-s)D_{\frac{1}{s}} (\bW_{{\btheta}}^n\cdot P_{T_P} \|Q_P^{(n)} ) },\Label{eq22}
\end{align}
where
$(a)$ follows from H\"{o}lder inequality.

Here, we choose $s$ to be $\frac{1}{n+1}$.
Since \eqref{3-28-1b} implies
\begin{align}
& (1-s)D_{\frac{1}{s}} (\bW_{{\btheta}}^n\cdot P_{T_P} \|Q_P^{(n)} ) 
=
(1-\frac{1}{n+1})D_{1+n} (\bW_{{\btheta}}^n\cdot P_{T_P} \|Q_P^{(n)} ) 
\nonumber\\
\le & (1-\frac{1}{n+1}) \frac{k}{2}\log n+O(1),
\Label{eq11b}
\end{align}
\eqref{eq22} yields 
\begin{align}
\lim_{n \to \infty} \frac{1}{M_n} \sum_{j=1}^{M_n} \sum_{i \neq j}
W^{n}_{{\btheta},E_{n,P,R}(i)}(\hat{{\cal D}}_{E_{n,P,R}(j)}) =0.\Label{eq23}
\end{align}
That is, the second term of \eqref{eq6} goes to zero.

\subsection{Evaluation of first term in \eqref{eq6}}
Now, we evaluate the first term in \eqref{eq6}, which is upper bounded by two terms in \eqref{eq7}.
So, we evaluate the first and the second terms in \eqref{eq7}.
For an arbitrary $\delta_2>0$,
we substitute $ \delta_2/\sqrt{n} $ into $\delta$ in \eqref{eq7}. 
So, \eqref{eq8T} implies that the second term in \eqref{eq7} goes to zero.
Further, for an arbitrary $\delta_3>0$, the first term in \eqref{eq7} is upper bounded as
\begin{align}
& W^{n}_{{\btheta},x^n}(
\{y^n| w^{n}_{{\btheta},x^n}(y^n) < 
e^{n R_1^*+ \sqrt{n}(R_2^*+\delta_2)} q_{P}^{(n)} (y^n) \}) 
\nonumber \\
\le & 
W^{n}_{{\btheta},x^n}(
\{y^n| w^{n}_{{\btheta},x^n}(y^n) < 
e^{n R_1^*+ \sqrt{n}(R_2^*+\delta_2+\delta_3)} w^{n}_{{\btheta},P}(y^n) \})
+
W^{n}_{{\btheta},x^n}(
\{y^n| w^{n}_{{\btheta},P}(y^n) < e^{-\sqrt{n} \delta_3 } q_{P}^{(n)} (y^n) \}) \Label{12-25-10}.
\end{align} 
We also have
\begin{align}
& W^{n}_{{\btheta},x^n}(
\{y^n| w^{n}_{{\btheta},P}(y^n) < e^{-\sqrt{n} \delta_3 } q_{P}^{(n)} (y^n) \}) \nonumber \\
=&
\sum_{{x^n}' \in T_P} \frac{1}{|T_P|}
W^{n}_{{\btheta},{x^n}'}(
\{y^n| w^{n}_{{\btheta},P}(y^n) < e^{-\sqrt{n} \delta_3 } q_{P}^{(n)} (y^n) \}) \nonumber \\
\le &
|T_n({\cal X})|
\sum_{{x^n}' \in {\cal X}^n} 
P^n(x)
W^{n}_{{\btheta},{x^n}'}(
\{y^n| w^{n}_{{\btheta},P}(y^n) < e^{-\sqrt{n} \delta_3 } q_{P}^{(n)} (y^n) \}) \nonumber \\
= &
|T_n({\cal X})|
W^{n}_{{\btheta},P}
(
\{y^n| w^{n}_{{\btheta},P}(y^n) < e^{-\sqrt{n} \delta_3 } q_{P}^{(n)} (y^n) \}) \nonumber \\
\le &
|T_n({\cal X})|
e^{-\sqrt{n} \delta_3 } Q_{P}^{(n)} 
(\{y^n| w^{n}_{{\btheta},P}(y^n) < e^{-\sqrt{n} \delta_3 } q_{P}^{(n)} (y^n) \}) \nonumber \\
\le &
|T_n({\cal X})|
e^{-\sqrt{n} \delta_3 } 
\to 0 \Label{12-25-11}.
\end{align} 
Therefore, due to \eqref{12-25-10} and \eqref{12-25-11},
the limit of the first term in \eqref{eq7} with 
$\delta= \delta_2/\sqrt{n}$
not larger than the limit of 
$\lim_{n \to \infty} 
W^{n}_{{\btheta},x^n}(
\{y^n| w^{n}_{{\btheta},x^n}(y^n) < 
e^{n R_1^*+ \sqrt{n}(R_2^*+\delta_2+\delta_3)} w^{n}_{{\btheta},P}(y^n) \})$.
So, \eqref{eq7} and \eqref{eq8T} guarantee that
\begin{align}
\lim_{n \to \infty} 
W^{n}_{{\btheta},x^n}(\hat{{\cal D}}_{x^n}^c)
\le 
\lim_{n \to \infty} 
W^{n}_{{\btheta},x^n}(
\{y^n| w^{n}_{{\btheta},x^n}(y^n) < 
e^{n R_1^*+ \sqrt{n}(R_2^*+\delta_2+\delta_3)} w^{n}_{{\btheta},P}(y^n) \})\label{12-25-T}
\end{align}
for arbitrary $\delta_2>0$ and $\delta_3>0$.

When $ I(P,W_{{\btheta}})>R_1^*$,
\begin{align}
W^{n}_{{\btheta},x^n}(
\{y^n| w^{n}_{{\btheta},x^n}(y^n) < e^{n R_1^*+ \sqrt{n}(R_2^*+\delta_2+\delta_3)} w^{n}_{{\btheta},P}(y^n) \})
\to 0
\Label{12-25-12B}.
\end{align} 
Now, we show that any element $x^n \in T_P$ satisfies that
\begin{align}
W^{n}_{{\btheta},x^n}(
\{y^n| w^{n}_{{\btheta},x^n}(y^n) < e^{n R_1^*+ \sqrt{n}(R_2^*+\delta_2+\delta_3)} w^{n}_{{\btheta},P}(y^n) \})
\to
\int_{-\infty}^{\frac{R_2^*+\delta_2+\delta_3}{\sqrt{V(P,\bW_{\btheta} )}}}
 \frac{1}{\sqrt{2\pi}}
\exp( - \frac{x^2}{2}) dx\Label{12-25-12}
\end{align} 
when $ I(P,\bW_{\btheta})=R_1^*$.
We have 
$\frac{1}{n} \log \frac{w^{n}_{{\btheta},x^n}(Y^n)}{w^{n}_{{\btheta},P}(Y^n) }
-R_1^*
=\frac{1}{n} \sum_{i=1}^n 
(\log \frac{w_{{\btheta},x_i}(Y_i)}{w_{{\btheta},P}(Y_i) }
- D(W_{\btheta,x_i}\| W_{\btheta,P} )
)$.
When $n$ goes to infinity, each element $x \in {\cal X}$ appears in $x^n \in T_P$ infinitely times.
So, the central limit theorem implies \eqref{12-25-12} with the assumptions of Theorem \ref{Th2},
So, when $ I(P,\bW_{\btheta})=R_1^*$,
combining \eqref{12-25-T}, \eqref{12-25-11}, and \eqref{12-25-12}, we have
\begin{align}
\lim_{n \to \infty} 
W^{n}_{{\btheta},x^n}(\hat{{\cal D}}_{x^n}^c)
\le
\int_{-\infty}^{\frac{R_2^*+\delta_2+\delta_3}{\sqrt{V(P,\bW_{\btheta} )}}}
 \frac{1}{\sqrt{2\pi}}
\exp( - \frac{x^2}{2}) dx. \Label{17-4}
\end{align} 
In summary, 
since $\delta_2$ and $\delta_3$ are arbitrary in \eqref{17-4}, 
we have
\begin{align}
\begin{array}{lll}
\lim_{n \to \infty} 
W^{n}_{{\btheta},x^n}(\hat{{\cal D}}_{x^n}^c)
&= 0 &\hbox{when } I(P, \bW_{{\btheta}})> R_1^*
\\
\lim_{n \to \infty} 
W^{n}_{{\btheta},x^n}(\hat{{\cal D}}_{x^n}^c)
&\le
\int_{-\infty}^{\frac{R_2^*}{\sqrt{V(P,\bW_{{\btheta}} )}}}
 \frac{1}{\sqrt{2\pi}}
\exp( - \frac{x^2}{2}) dx&\hbox{when } I(P,\bW_{{\btheta}})=R_1^*.
\end{array}
\Label{eq24}
\end{align} 
Hence, the relations \eqref{eq6}, \eqref{eq24}, and \eqref{eq23}
yield \eqref{eq2-2}.
Therefore, we complete our proof of Theorem \ref{Th2}.

\subsection{Proof of Theorem \ref{Th3}}
Now, we proceed to the proof of Theorem \ref{Th3}.
Similarly, we evaluate the average error probability by discussing the first and second terms of \eqref{eq6}, separately.
Fortunately, the upper bound in the RHS of \eqref{eq22} 
goes to zero uniformly with respect to ${\btheta}$ in any compact set 
because the constant $O(1)$ in \eqref{3-28-1b} can be chosen uniformly with respect to ${\btheta}$ in any compact set. 
So, the second term of \eqref{eq6} goes to zero uniformly with respect to ${\btheta}$ in any compact set.

Now, we evaluate the first term of \eqref{eq6}, which is upper bounded by two terms in \eqref{eq7}.
Due to \eqref{eq8T}, the second term of \eqref{eq7} goes to zero uniformly with respect to ${\btheta}$ in any compact set because the term $o(1)$ in \eqref{eq8T} can be chosen as a arbitrary small constant uniformly with respect to ${\btheta}$ in any compact set. 
Then, the remaining term is the first term of \eqref{eq7},
which is upper bounded by two terms in \eqref{12-25-10}.
The second term of \eqref{12-25-10} goes to zero due to \eqref{12-25-11}.
Therefore, it is enough to evaluate 
$\lim_{n \to \infty} 
W^{n}_{{\btheta},x^n}(
\{y^n| w^{n}_{{\btheta},x^n}(y^n) < 
e^{n R_1^*+ \sqrt{n}(R_2^*+\delta_2+\delta_3)} w^{n}_{{\btheta},P}(y^n) \})$, which is the first term of \eqref{12-25-10}.

Since the function ${\btheta} \mapsto I(P,\bW_{\btheta} )$ is a $C^1$ function on $\bTheta$
and the function ${\btheta} \mapsto V(P,\bW_{\btheta} )$ is a continuous function on $\bTheta$,
the likelihood ratio
$\frac{1}{n} \log \frac{w^{n}_{{\btheta},x^n}(Y^n)}{w^{n}_{{\btheta},P}(Y^n) }$
has the expectation 
$I(P,\bW_{{\btheta}_1})+\frac{1}{\sqrt{n}} f({\btheta}_2)+o(\frac{1}{\sqrt{n}})$
and the variance $\frac{1}{n}V(P,\bW_{\btheta} )+o(\frac{1}{n}) $
with ${\btheta}={\btheta}_1 +\frac{1}{\sqrt{n}}{\btheta}_2$ 
for $x^n \in T_P$.
This property yields 
the relation
\begin{align}
&\lim_{n \to \infty}
w^{n}_{{\btheta},x^n}(
\{y^n| w^{n}_{{\btheta},x^n}(y^n) < e^{n R_1^*+ \sqrt{n}R_2^*} w^{n}_{{\btheta},P}(y^n) \})\nonumber  \\
=&
\left\{
\begin{array}{ll}
0 
& \hbox{when } I(P,\bW_{{\btheta}_1} ) > R_1^* \\
\displaystyle \int_{-\infty}^{\frac{R_2^*-f({\btheta}_2)}{\sqrt{V(P,\bW_{{\btheta}_1} )}}}
 \frac{1}{\sqrt{2\pi}}
\exp( - \frac{x^2}{2}) dx
& \hbox{when } I(P,\bW_{{\btheta}_1} ) = R_1^* 
\end{array}
\right.
\Label{12-25-12x}
\end{align} 
for $x^n \in T_P$.
The convergence is uniform with respect to ${\btheta}_2$ in any compact set. 
Since $\delta_2$ and $\delta_3$ are arbitrary, we obtain the desired argument.
Since the uniformity in any compact subset has been discussed in the above discussion.
the proof of Theorem \ref{Th3} is completed, now.


\section{Conclusion}\Label{s7}
We have proposed 
a universal channel coding for general output alphabet
including continuous output alphabets.
Although our encoder is the same as the encoder given in the previous paper \cite{Ha1},
we cannot directly apply the decoder given in \cite{Ha1}
because 
it is not easy to make a distribution that universally approximates 
any independent and identical distribution in the sense of maximum relative entropy in the continuous alphabet.
To overcome the difficulty,
we have invented an $\alpha$-R\'{e}nyi-relative-entropy version of Clarke and
Barron's formula for Bayesian average distribution.
That is, we have shown that
the Bayesian average distribution well approximates 
any independent and identical distribution in the sense of 
$\alpha$-R\'{e}nyi relative entropy.
Then, we have made our universal decoder by applying the information spectrum method
to the Bayesian average distribution.
We have lower bounded the error exponent of our universal code,
which implies that our code attains the mutual information rate.
Since our approach covers the discrete and continuous cases and the exponential and the second-order type evaluations for 
the decoding error probability,
our method provides a unified viewpoint for the universal channel coding,
which is an advantage over the existing studies \cite{CK,de4,C,D,E,F,G,H,Do}.

Further, we have introduced the parametrization 
${\btheta}_1+\frac{{\btheta}_2}{\sqrt{n}}$ for our channel,
which is commonly used in statistics for 
discussing the asymptotic local approximation by normal distribution family 
\cite{Lehmann,Vaart}.
This parametrization matches the second order parameterization of the coding rate.
So, we can expect that 
this parametrization is applicable to the case when we have an unknown small disturbance in the channel.

Here, we compare our analysis on universal coding with the compound channel \cite{BBT,Wolfowitz,G}.
In the compound channel, we focus on the worst case of the average error 
probability for the unknown channel parameter.
Hence, we do not evaluate the average error 
probability when the channel parameter is not the worst case.
However, in the universal coding \cite{CK}, we evaluate the error probability for all possible channels.
Hence, we can evaluate how better the average error probability of each case is than the worst case.
In particular, our improved second order analysis in Theorem \ref{Th3} 
clarifies its dependence of the unknown parameter ${\btheta}_2$.

\section*{Acknowledgment}
The author is grateful for Dr. Hideki Yagi to 
explaining MIMO Gaussian channels
and informing the references \cite{de4,A,B,C,D,E,F,G,H}. 
He is grateful to the referee for Dr. Vincent Tan for helpfu comments.
He is also grateful to the referee of the first version of this paper for 
helpful comments.
The author is partially supported by JSPS Grants-in-Aid for Scientific Research (A) No. 23246071
and (A) No.17H01280.
The author is also partially supported by the National Institute of Information and Communication Technology (NICT), Japan.

\appendices
\section{Proof of Lemma \ref{L11-27-1}}\Label{A1}
When $\inf_{{\btheta}\in \bTheta_0}I(P,\bW_{\btheta} ) \le R$,
all terms in \eqref{eq11-27-1} are zero.
So, we can assume that 
$\inf_{{\btheta}\in \bTheta_0}I(P,\bW_{\btheta} ) > R$ without loss of generality.

Firstly, we show that
\begin{align}
&\max_{R_1} \min (\max_{s \in [0,1]}(s I_{1-s}(P,\bW_{\btheta} ) -s R_1),R_1-R) 
\nonumber \\
=&\max_{s\in [0,1]}\frac{1}{1+s}(sI_{1-s}(P,\bW_{\btheta} ) -sR) .\Label{2-15-1}
\end{align}
Since the function $R_1 \mapsto \max_{s \in [0,1]}(sI_{1-s}(P,\bW_{\btheta} ) -sR_1)$ is monotone decreasing and continuous 
and
the function $R_1 \mapsto R_1-R$ is 
monotone increasing and continuous, 
there exists a real number $R_1^*>R$ such that
$\max_{s \in [0,1]}(sI_{1-s}(P,\bW_{\btheta} ) -s R_1^*) =R_1^*-R$.
We also choose $s^*:=\argmax _{s \in [0,1]}(sI_{1-s}(P,\bW_{\btheta} ) -sR_1^*)$.
Here, we assume that $s^*\in (0,1)$.
$s^*$ satisfies $\frac{d s I_{1-s}(P,\bW_{\btheta} )}{ds}|_{s=s^*}=R_1^*$.
Since $(s^* I_{1-s^*}(P,\bW_{\btheta} ) -s^* R_1^*) =R_1^*-R$,
we have
$R_1^*= \frac{R-s^* I_{1-s^*}(P,\bW_{\btheta} )}{1+s^*}$.
So, $\max_{s \in [0,1]}(s I_{1-s}(P,\bW_{\btheta} ) -s R_1^*) =
\frac{ (s^* I_{1-s^*}(P,\bW_{\btheta} )-s^* R)}{1+s^*}$
and
$\frac{d s I_{1-s}(P,\bW_{\btheta} )}{ds}|_{s=s^*}=\frac{R-s^* I_{1-s^*}(P,\bW_{\btheta} )}{1+s^*}$.

Since the first derivative of 
$\frac{ (s I_{1-s}(P,\bW_{\btheta} )-s R)}{1+s}$ with respect to $s$
is 
$\frac{(1+s)\frac{d s I_{1-s}(P,\bW_{\btheta} )}{ds} -s I_{1-s}(P,\bW_{\btheta} )-R}{(1+s)^2}$
and 
$(1+s)\frac{d s I_{1-s}(P,\bW_{\btheta} )}{ds} -s I_{1-s}(P,\bW_{\btheta} )-R$
is monotone decreasing for $s \in [0,1]$,
$s_*:=\argmax_{s \in [0,1]}
\frac{ (s I_{1-s}(P,\bW_{\btheta} )-s R)}{1+s}$ 
satisfies the same condition
$\frac{d s I_{1-s}(P,\bW_{\btheta} )}{ds}|_{s=s_*}=\frac{R-s_* I_{1-s_*}(P,\bW_{\btheta} )}{1+s_*}$.
So, we find that
$\max_{s \in [0,1]}\frac{ (s I_{1-s}(P,\bW_{\btheta} )-s R)}{1+s}=
\frac{ (s^* I_{1-s^*}(P,\bW_{\btheta} )-s^* R)}{1+s^*}$.
Thus, we obtain \eqref{2-15-1} when $s^*\in (0,1)$.
When $s^*=0$, we can show $s_*=0$, which implies \eqref{2-15-1}.
Similarly, we can show \eqref{2-15-1} when $s^*=1$.

Since 
\begin{align}
\max_{R_1} \inf_{{\btheta}\in \bTheta_0}\min (\max_{s \in [0,1]}(s I_{1-s}(P,\bW_{\btheta} ) -s R_1),R_1-R)
\le
\inf_{{\btheta}\in \bTheta_0}\max_{R_1} \min (\max_{s \in [0,1]}(s I_{1-s}(P,\bW_{\btheta} ) -s R_1),R_1-R) ,
\end{align}
it is sufficient to show there exists $R_1$ such that
\begin{align}
\inf_{{\btheta}\in \bTheta_0}\min (\max_{s \in [0,1]}(sI_{1-s}(P,\bW_{\btheta} ) -s R_1),R_1-R)
\ge
\inf_{{\btheta}\in \bTheta_0} \max_{s\in [0,1]}\frac{1}{1+s}(sI_{1-s}(P,\bW_{\btheta} ) -s R) .\Label{2-15-2}
\end{align}
We choose $R_1$ to be $R+  \inf_{{\btheta}\in \bTheta_0} \max_{s\in [0,1]}\frac{1}{1+s}(sI_{1-s}(P,\bW_{\btheta} ) -sR) $.
Given a parameter ${\btheta} \in \bTheta_0$,
using the function $f(s,{\btheta}):=\frac{1}{1+s}(sI_{1-s}(P,\bW_{\btheta} ) -sR)$
and $s_{{\btheta}}:= \argmax_{s \in [0,1]} f(s,{\btheta})$, 
we have
\begin{align}
f(s_{{\btheta}},{\btheta})
\ge
\inf_{{\btheta}'\in \bTheta_0} \max_{s'\in [0,1]} f(s',{\btheta}'),
\end{align}
which implies that
\begin{align}
&
\max_{s \in [0,1]}(sI_{1-s}(P,\bW_{\btheta} ) -sR_1)
\ge 
(s_{{\btheta}} I_{1-s_{{\btheta}}}(P,\bW_{\btheta} ) -s_{{\btheta}} R_1)
\nonumber \\
=&
(s_{{\btheta}} I_{1-s_{{\btheta}}}(P,\bW_{\btheta} ) -s_{{\btheta}} R)+ s_{{\btheta}} 
\inf_{{\btheta}'\in \bTheta_0} \max_{s'\in [0,1]}f(s',{\btheta}')
\nonumber \\
=&
f(s_{{\btheta}},{\btheta}) + s_{{\btheta}} 
(f(s_{{\btheta}},{\btheta})-\inf_{{\btheta}'\in \bTheta_0} \max_{s'\in [0,1]} f(s',{\btheta}'))
\ge f(s_{{\btheta}},{\btheta})
\nonumber \\
\ge &
\inf_{{\btheta}\in \bTheta_0} \max_{s\in [0,1]}\frac{1}{1+s}(sI_{1-s}(P,\bW_{\btheta} ) -sR)
=R_1-R.
\Label{2-15-3}
\end{align}
Thus,
\begin{align}
\min (\max_{s \in [0,1]}(sI_{1-s}(P,\bW_{\btheta} ) -sR_1),R_1-R)
=
\inf_{{\btheta}\in \bTheta_0} 
\max_{s\in [0,1]}\frac{1}{1+s}(sI_{1-s}(P,\bW_{\btheta} ) -sR),
\end{align}
which implies \eqref{2-15-2}.
\endproof

\section{Proof of $\alpha$-R\'enyi relative entropy version of Clarke-Barron formula}\Label{A2}
\subsection{Preparation}
To show Lemma \ref{l2-A}, we prepare several formulas used in proofs of \eqref{3-28-Ab} and \eqref{3-28-A}.
Under the assumption of exponential family, 
we consider 
the logarithmic derivatives 
$l_{{\btheta}}(y):=(l_{{\btheta},1}(y), \ldots, l_{{\btheta},k}(y))$, where
\begin{align}
l_{{\btheta},j}(y)
:=\frac{\partial}{\partial {\theta}^j} \log p_{{\btheta}}(y) 
= g_{j}(y)-\frac{\partial}{\partial {\theta}^j} \phi({\btheta}).
\end{align}
The Fisher information matrix $J_{{\btheta},i,j}$
is given as 
\begin{align}
J_{{\btheta},i,j}
:=
\frac{\partial^2}{\partial \theta^i \partial \theta^j} \phi(\btheta)
=
\int_{{\cal Y}}
\frac{\partial}{\partial \theta^i} l_{\btheta,j}(y)
p_{\btheta}(y)\mu(dy).
\end{align}
Due to the above assumption, $J_{\btheta,i,j}$ is continuous for $\btheta\in \bTheta$. 
Hence, we have
\begin{align}
\frac{\partial^2}{\partial \theta^i \partial \theta^j} \log p_{\btheta}(y) 
=-J_{\theta,i,j}
\end{align}
which is independent of $y$.

Define 
\begin{align}
l_{\btheta,j;n}(y^n) &:=
\frac{1}{\sqrt{n}}
\sum_{i=1}^n \frac{\partial}{\partial \theta^j}
\log p_{\btheta}(y_i) .
\end{align}
Since $p_{\btheta}$ is an exponential family, we define
\begin{align}
J_{\btheta,j,j';n} :=
-\frac{1}{n}
\sum_{i=1}^n \frac{\partial^2}{\partial\theta^j \partial\theta^{j'}}
\log p_{\btheta}(y_i),\Label{Eq222}
\end{align}
which is also independent of for $y^n \in {\cal Y}^n$.

In the following discussion, we employ the 
Laplace approximation (Laplace method of approximation).
To use the Taylor expansion at $\btheta_0$, 
we choose $\btheta_1(\btheta) $ between $\btheta $ and $\btheta_0$.
Then, we have
\begin{align}
& 
\frac{p_{\btheta}^n(y^n)}{p_{\btheta_0}^n(y^n)}
\nu(\btheta) \nonumber \\
=& 
\frac{e^{\sum_{i=1}^n \log p_{\btheta}(y_i)}
}{p_{\btheta_0}^n(y^n)}
\nu(\btheta)  \nonumber \\
=& 
\frac{e^{
\sum_{i=1}^n \log p_{\btheta_0}(y_i)
+ \sum_{j=1}^k (\theta^j-\theta_0^j) 
\sum_{i=1}^n \frac{\partial}{\partial \theta^j}
\log p_{\btheta}(y_i)|_{\btheta=\btheta_0}
+ 
\sum_{j,j'=1}^k
\frac{(\theta^j -\theta_0^j )(\theta^{j'} -\theta_0^{j'})}{2}
\sum_{i=1}^n \frac{\partial^2}{\partial\theta^j \partial\theta^{j'}}
\log p_{\btheta}(y_i)|_{\theta=\theta_1}
}
}{p_{\btheta_0}^n(y^n)}
\nu(\btheta) \nonumber \\
=& 
e^{
-n (\btheta-\btheta_0)^T \frac{J_{\btheta_1(\btheta)}}{2}(\btheta-\btheta_0)
+\sqrt{n} (\btheta-\btheta_0)^T l_{\btheta_0;n} (y^n)
}
\nu(\btheta) ,\Label{12-25-9}
\end{align}
where the final equation follows from the properties of 
$l_{\btheta;n}(y^n):=(l_{\btheta,j;n}(y^n))$ and the matrix $J_{\btheta}$.

Next, for an arbitrary $\epsilon>0$,
we choose a neighborhood $
U_{\btheta_0,\delta}:=\{\btheta |  \|\btheta-\btheta_0\| \le \delta\} $
such that
$J_{\btheta} \le J_{\btheta_0} (1+\epsilon)$ and $\nu(\btheta) \ge \nu(\theta_0) (1-\epsilon)$
for $\btheta \in U_{\btheta_0,\delta}$.
For $\btheta \in U_{\btheta_0,\delta}$, we have
\begin{align}
& 
e^{
-n (\btheta-\btheta_0)^T \frac{J_{\btheta_1(\btheta)}}{2}(\btheta-\btheta_0)
+\sqrt{n} (\btheta-\btheta_0)^T l_{\theta_0;n} (y^n)
}
\nu(\btheta) \nonumber \\
\ge & 
e^{
-n (\btheta-\btheta_0)^T \frac{J_{\btheta_0}(1+\epsilon)}{2}(\btheta-\btheta_0)
+\sqrt{n} (\btheta-\btheta_0)^T l_{\btheta_0;n} (y^n)
}
\nu(\btheta_0) (1-\epsilon) \Label{12-25-8}\\
= &
e^{\frac{1}{2}
l_{\btheta_0;n}(y^n)^T
J_{\btheta_0}^{-1} (1+\epsilon)^{-1} 
l_{\btheta_0;n}(y^n)  }
e^{
-
(\sqrt{n} (\btheta-\btheta_0 ) - (J_{\btheta_0}(1+\epsilon ))^{-1} l_{\btheta_0;n}(y^n))^T
\frac{J_{\btheta_0}(1+\epsilon )}{2}
(\sqrt{n} (\btheta-\btheta_0 ) - (J_{\btheta_0}(1+\epsilon ))^{-1} l_{\btheta_0;n}(y^n))
}
\nu(\btheta_0) (1-\epsilon) .\Label{12-25-7}
\end{align}

\subsection{Proof of \eqref{3-28-Ab}}
Now, we prove \eqref{3-28-Ab}.  
We focus on the set
$B_{y^n}:=\{z \in \mathbb{R}^k| z^T l_{\btheta_0;n} (y^n) \ge 0,
\|z\|\le 1\}$ for $y^n$.
For $n \ge \frac{1}{\delta^2}$,
we have
\begin{align}
& \frac{q_{\nu}^n(y^n)}{p_{\btheta_0}^n(y^n)}
=\int_{\bTheta} 
\frac{p_{\btheta}^n(y^n)}{p_{\btheta_0}^n(y^n)}
\nu(\btheta) d \btheta \nonumber \\
\ge & \int_{U_{\btheta_0,\delta}} 
\frac{p_{\btheta}^n(y^n)}{p_{\btheta_0}^n(y^n)}
\nu(\btheta) d \btheta  \nonumber \\
\stackrel{(a)}{\ge} & 
\int_{U_{\btheta_0,\delta}} 
e^{
-n (\btheta-\theta_0)^T \frac{J_{\btheta_0}(1+\epsilon)}{2}(\btheta-\btheta_0)
+\sqrt{n} (\btheta-\btheta_0)^T l_{\btheta_0;n} (y^n)
}
\nu(\btheta_0) (1-\epsilon) 
d \btheta  \nonumber \\
\stackrel{(b)}{=} & 
\frac{1}{n^{\frac{k}{2}}}
\int_{\| z\|\le \sqrt{n} \delta} 
e^{
-z^T \frac{J_{\btheta_0}(1+\epsilon)}{2}z
+z^T l_{\btheta_0;n} (y^n)
}
\nu(\btheta_0) (1-\epsilon) 
d z  \nonumber \\
\stackrel{(c)}{\ge} & 
\frac{1}{n^{\frac{k}{2}}}
\int_{B_{y^n}} 
e^{
-z^T \frac{J_{\btheta_0}(1+\epsilon)}{2}z
+z^T l_{\btheta_0;n} (y^n)
}
\nu(\btheta_0) (1-\epsilon) 
d z  \nonumber \\
\stackrel{(d)}{\ge} & 
\frac{1}{n^{\frac{k}{2}}}
e^{-\frac{ \|J_{\btheta_0}\|(1+\epsilon)}{2}}
\int_{B_{y^n}} 
\nu(\btheta_0) (1-\epsilon) 
d z  \nonumber \\
\stackrel{(e)}{\ge} & 
\frac{1}{n^{\frac{k}{2}}}
e^{-\frac{ \|J_{\btheta_0}\|(1+\epsilon)}{2}}
\frac{\pi^{\frac{k}{2}}}{2\Gamma(\frac{k}{2}+1)}
 (1-\epsilon) ,
\end{align}
where $(a)$, $(b)$, $(c)$, and $(d)$
follow from 
\eqref{12-25-9} and \eqref{12-25-8},
the relation $z= \sqrt{n} (\btheta-\btheta_0)$,
the relation $n \ge \frac{1}{\delta^2}$,
and
the relation $ \|J_{\btheta_0} \| \ge z^T J_{\btheta_0}z$
for $\|z\| \le 1$, respectively.
The inequality $(e)$ is shown because the volume of $B_{y^n}$ is $\frac{\pi^{\frac{k}{2}}}{2\Gamma(\frac{k}{2}+1)}$.
That is, we have
\begin{align}
\frac{1}{n^{\frac{k}{2}}}\cdot 
\frac{p_{\btheta_0}^n(y^n)}
{q_{\nu}^n(y^n)}
\le 
e^{\frac{ \|J_{\btheta_0}\|(1+\epsilon)}{2}}
\frac{2\Gamma(\frac{k}{2}+1)}{( 1-\epsilon ) \pi^{\frac{k}{2}}}
\Label{12-25-2}.
\end{align}

Therefore, for $n \ge \frac{1}{\delta^2}$,
using \eqref{12-25-2} we have
\begin{align}
& n^{-s\frac{k}{2}}
e^{s D_{1+s}(P_{\btheta_0}^n\|Q_{\nu}^n)}\nonumber \\
=&
 n^{-s\frac{k}{2}}
E_{P_{\btheta_0}^n}
\bigg[\bigg(\frac{p_{\btheta_0}^n(Y^n)}{q_{\nu}^n(Y^n)}\bigg)^s\bigg]
\nonumber \\
\le &
E_{P_{\btheta_0}^n}
\bigg[
e^{s\frac{ \|J_{\btheta_0}\|(1+\epsilon)}{2}}
\bigg(\frac{2\Gamma(\frac{k}{2}+1)}{( 1-\epsilon ) \pi^{\frac{k}{2}}}\bigg)^s
\bigg]
\nonumber \\
\le &
e^{s\frac{ \|J_{\btheta_0}\|(1+\epsilon)}{2}}
\bigg(\frac{2\Gamma(\frac{k}{2}+1)}{( 1-\epsilon ) \pi^{\frac{k}{2}}}\bigg)^s.\Label{16-b}
\end{align}
Substituting $n$ into $s$ in \eqref{16-b}, 
we have
\begin{align*}
D_{1+n}(P_{\btheta}^n\|Q_{\nu}^n) - \frac{k}{2}\log n
=
\frac{1}{n}\log  n^{-n\frac{k}{2}}
e^{n D_{1+n}(P_{\btheta_0}^n\|Q_{\nu}^n)}
\le
\log \bigg[e^{\frac{ \|J_{\btheta_0}\|(1+\epsilon)}{2}}
\Big(\frac{2\Gamma(\frac{k}{2}+1)}{( 1-\epsilon ) \pi^{\frac{k}{2}}}\Big)\bigg].
\end{align*}
Since $\log \bigg[e^{\frac{ \|J_{\btheta_0}\|(1+\epsilon)}{2}}
\Big(\frac{2\Gamma(\frac{k}{2}+1)}{( 1-\epsilon ) \pi^{\frac{k}{2}}}\Big)\bigg]$ is a constant,
we obtain \eqref{3-28-Ab}.

\subsection{Proof of \eqref{3-28-A}}
When the continuity in the assumption is uniform for $\theta$ in any compact subset,
we can choose a common constant $\delta>0$ in a compact subset in $\bTheta$.
Since the constant $\delta$ decides the range for $n$ in the above discussion,
the constant on the RHS of \eqref{3-28-Ab} can be chosen uniformly
for $\theta$ in any compact subset.

Next, for a deeper analysis for proving \eqref{3-28-A}, 
we fix an arbitrary small real number $\epsilon>0$.
Then, we choose a sufficiently large real number $R$ 
and a large integer $N_1$ such that
the complement $C_n^c$ of the set $C_n:= \{ y^n | \|l_{\btheta_0;n}(y^n) \| < R\}$
satisfies
\begin{align}
P_{\btheta_0} (C_n^c) \le \epsilon\Label{12-25-3}
\end{align}
for $n \ge N_1$.
Then, we evaluate 
$\frac{q_{\nu}^n(y^n)}{p_{\btheta_0}^n(y^n)}$
under the assumption $\|l_{\btheta_0;n}(y^n) \|< R$ as follows.
Then, we can choose sufficiently large $N_2$ such that
\begin{align}
& \int_{ \|z\| \le \sqrt{n} \delta} 
e^{
-
(z - (J_{\btheta_0}(1+\epsilon ))^{-1} l_{\btheta_0;n}(y^n))^T
\frac{J_{\btheta_0}(1+\epsilon )}{2}
(z - (J_{\btheta_0}(1+\epsilon ))^{-1} l_{\btheta_0;n}(y^n))
}
d z \nonumber \\
\ge &
\frac{ (2\pi)^{\frac{k}{2}}}{ (\det (J_{\btheta_0}))^{\frac{1}{2}}
(1+\epsilon )^{\frac{k}{2}}}
(1-\epsilon) \Label{12-25-a}
\end{align}
for $n \ge N_2$ and $y^n$ satisfying 
that $\|l_{\btheta_0;n}(y^n) \|< R$
because the limit of LHS of \eqref{12-25-a} is $\frac{ (2\pi)^{\frac{k}{2}}}{ (\det (J_{\btheta_0}))^{\frac{1}{2}}
(1+\epsilon )^{\frac{k}{2}}}$.
Thus, when $n \ge N_2$, we have
\begin{align}
& \frac{ p_{\nu}^n(y^n)}{ p_{\btheta_0}^n(y^n)}
=\int_{\bTheta} 
\frac{ p_{\btheta}^n(y^n)}{ p_{\btheta_0}^n(y^n)}
w(\btheta) d \btheta \nonumber \\
\ge & \int_{U_{\btheta_0,\delta}} 
\frac{ p_{\btheta}^n(y^n)}{ p_{\btheta_0}^n(y^n)}
\nu(\btheta) d \btheta \nonumber \\
\stackrel{(a)}{\ge} & 
e^{\frac{1}{2}
l_{\btheta_0;n}(y^n)^T
J_{\btheta_0}^{-1} (1+\epsilon)^{-1} 
l_{\btheta_0;n}(y^n)  } \nonumber\\
& \cdot \int_{U_{\btheta_0,\delta}} 
e^{
-
(\sqrt{n} (\btheta-\btheta_0 ) - (J_{\btheta_0}(1+\epsilon ))^{-1} l_{\btheta_0;n}(y^n))^T
\frac{J_{\btheta_0}(1+\epsilon )}{2}
(\sqrt{n} (\btheta-\btheta_0 ) - (J_{\btheta_0}(1+\epsilon ))^{-1} l_{\btheta_0;n}(y^n))
}
\nu(\btheta_0) (1-\epsilon) 
d \theta \nonumber \\
\stackrel{(b)}{=} & 
n^{-\frac{k}{2}}
 e^{\frac{1}{2}
l_{\btheta_0;n}(y^n)^T
J_{\btheta_0}^{-1} (1+\epsilon)^{-1} 
l_{\btheta_0;n}(y^n)  }\nonumber \\
& \cdot \int_{ \|z\| \le \sqrt{n} \delta} 
e^{
-
(z - (J_{\btheta_0}(1+\epsilon ))^{-1} l_{\btheta_0;n}(y^n))^T
\frac{J_{\btheta_0}(1+\epsilon )}{2}
(z - (J_{\btheta_0}(1+\epsilon ))^{-1} l_{\btheta_0;n}(y^n))
}
\nu(\btheta_0) (1-\epsilon) 
d z \nonumber \\
\stackrel{(c)}{\ge} & 
n^{-\frac{k}{2}}
 e^{\frac{1}{2}
l_{\btheta_0;n}(y^n)^T
J_{\btheta_0}^{-1} (1+\epsilon)^{-1} 
l_{\btheta_0;n}(y^n)  } 
\nu(\btheta_0) (1-\epsilon)^2
\frac{ (2\pi)^{\frac{k}{2}}}{ (\det (J_{\btheta_0}))^{\frac{1}{2}}
(1+\epsilon )^{\frac{k}{2}}}\Label{12-25-1},
\end{align}
where 
$(a)$, $(b)$, and $(c)$ follow from 
\eqref{12-25-9} and \eqref{12-25-7}, the relation $z= \sqrt{n} (\btheta-\btheta_0)$, 
and \eqref{12-25-a}.

Now, we introduce a notation.
For a distribution $P$, a subset $S$ of the probability space,
and a random variable $X$, 
we denote the value
$\int_{S} X(\omega) P(d \omega)$ by
$E_{P|S}[X]$.
Therefore, using this notation,
we have
\begin{align}
& n^{-s\frac{k}{2}}
e^{s D_{1+s}(P_{\btheta_0}^n\|Q_{\nu}^n)}\nonumber \\
=&
 n^{-s\frac{k}{2}}
E_{P_{\btheta_0}^n}
\Bigl[ \Big(\frac{p_{\btheta_0}^n(Y^n)}{q_{\nu}^n(Y^n)}\Big)^s
\Bigr]
\nonumber \\
= &
 n^{-s\frac{k}{2}}
E_{P_{\btheta_0}^n|C_n}
\Bigl[\Big(\frac{p_{\btheta_0}^n(Y^n)}{q_{\nu}^n(Y^n)}\Big)^s\Bigr]
+
 n^{-s\frac{k}{2}}
E_{P_{\btheta_0}^n|C_n^c}
\Bigl[\Big(\frac{p_{\btheta_0}^n(Y^n)}{q_{\nu}^n(Y^n)}\Big)^s\Bigr]
\nonumber \\
\stackrel{(a)}{\le} &
E_{P_{\btheta_0}^n|C_n}
\Bigl[
e^{-\frac{s}{2}
l_{\btheta_0;n}(y^n)^T
J_{\btheta_0}^{-1} (1+\epsilon)^{-1} 
l_{\btheta_0;n}(y^n)  } 
\nu(\btheta_0)^{-s} (1-\epsilon)^{-2s}
\frac
{ (\det (J_{\btheta_0}))^{\frac{s}{2}}(1+\epsilon )^{\frac{sk}{2}}}
{ (2\pi)^{\frac{sk}{2}}}
\Bigr]
\nonumber \\
& +
E_{P_{\btheta_0}^n|C_n^c}
\Bigl[
e^{s\frac{ \|J_{\btheta_0}\|(1+\epsilon)}{2}}
\Big(\frac{2\Gamma(\frac{k}{2}+1)}{( 1-\epsilon ) \pi^{\frac{k}{2}}}\Big)^s
\Bigr]
\nonumber \\
\stackrel{(b)}{\le} &
E_{P_{\btheta_0}^n}
\biggl[
e^{-(1+\epsilon)^{-1} 
\frac{s}{2}
l_{\btheta_0;n}(y^n)^T
J_{\btheta_0}^{-1} 
l_{\btheta_0;n}(y^n)  } 
\nu(\btheta_0)^{-s} (1-\epsilon)^{-2s}
\frac
{ (\det (J_{\btheta_0}))^{\frac{s}{2}}(1+\epsilon )^{\frac{sk}{2}}}
{ (2\pi)^{\frac{sk}{2}}}
\biggr]
\nonumber \\
& +
e^{s\frac{ \|J_{\btheta_0}\|(1+\epsilon)}{2}}
\Big(\frac{2\Gamma(\frac{k}{2}+1)}{( 1-\epsilon ) \pi^{\frac{k}{2}}}\Big)^s
\epsilon,
\end{align}
where $(a)$ and $(b)$ follow from 
\eqref{12-25-2}, \eqref{12-25-1} and \eqref{12-25-3}, respectively.

Due to the central limit theorem, the random variable 
$ l_{\btheta_0;n}(Y^n)^T J_{\btheta_0}^{-1} l_{\btheta;n}(Y^n)$
asymptotically obeys the 
$\chi^2$ distribution with the degree $k$.
Hence, 
the expectation
$E_{P_{\btheta_0}^n}
[e^{-(1+\epsilon)^{-1} \frac{s}{2} 
l_{\btheta_0;n}(Y^n)^T J_{\btheta_0}^{-1} l_{\btheta_0;n}(Y^n)}]$
converges to $\frac{1}{(1+s(1+\epsilon)^{-1})^{k/2}}
$.
Thus,
\begin{align}
& \lim_{n \to \infty} n^{-s\frac{k}{2}}
e^{s D_{1+s}(P_{\btheta_0}^n\|Q_{\nu}^n)}\nonumber \\
\le &
\frac{1}{(1+s(1+\epsilon)^{-1})^{k/2}}
\nu(\btheta_0)^{-s} (1-\epsilon)^{-2s}
\frac
{ (\det (J_{\btheta_0}))^{\frac{s}{2}}(1+\epsilon )^{\frac{sk}{2}}}
{ (2\pi)^{\frac{sk}{2}}}
\nonumber \\
& +
e^{s\frac{ \|J_{\btheta_0}\|(1+\epsilon)}{2}}
\Big(\frac{2\Gamma(\frac{k}{2}+1)}{( 1-\epsilon ) \pi^{\frac{k}{2}}}\Big)^s
\epsilon.
\end{align}
Since $\epsilon>0$ is arbitrary, we have
\begin{align}
& \lim_{n \to \infty} n^{-s\frac{k}{2}}
e^{s D_{1+s}(P_{\theta_0}^n\|Q_{\nu}^n)}\nonumber \\
\le &
\frac{1}{(1+s)^{k/2}}
\nu(\btheta_0)^{-s} 
\frac
{ (\det (J_{\btheta_0}))^{\frac{s}{2}}}
{ (2\pi)^{\frac{sk}{2}}}.
\end{align}
Taking the logarithm, we have
\begin{align}
& \lim_{n \to \infty} 
s D_{1+s}(P_{\theta_0}^n\|Q_{\nu}^n)
-s\frac{k}{2} \log n
\le 
-\frac{k}{2}\log (1+s)
-s \log \nu(\btheta_0) 
+ \frac{s}{2}\log \det (J_{\btheta_0})
- \frac{sk}{2} \log   (2\pi),\Label{Eq223}
\end{align}
which implies \eqref{3-28-A}.
\endproof

\section{Proof of Lemma \ref{l1}}\Label{Ap3}
To prove Lemma \ref{l1}, we recall packing lemma \cite[Lemma 10.1]{CK}.
We assume that ${\cal X}$ and ${\cal Y}$ are finite sets.
For a type $P\in T_n({\cal X})$, 
let ${\cal V}(P, {\cal Y})$ be the set of conditional types of $P$ from ${\cal X}$ to ${\cal Y}$.
Then, we have the following proposition.

\begin{proposition}[\protect{\cite[Lemma 10.1]{CK}}]\Label{Pack}
For every $R>\delta>0$, there exists a sufficiently large integer $N$ satisfying the following.
For every $n \ge N$ and every type $P \in T_n({\cal X})$ satisfying $H(P)>R$, 
there exist at least $\exp[n(R-\delta)]$
distinct sequences $x_i^n \in T_p$ such that every pair of stochastic matrices 
$\bV,\hat{\bV} \in {\cal V}(P, {\cal Y})$
and every $i$ satisfy the inequality
\begin{align}
\bigg| T_{\bV}(x_i^n) \cap \bigg( \bigcup_{j\neq i} T_{\hat{\bV}}(x_j^n)\bigg) \bigg|
\le | T_{{\bV}}(x_i^n)\Big|
\exp [ -n |  I(P, \hat{\bV}) -R|^+]. \Label{Eq5-24}
\end{align}
\end{proposition}

\begin{proofof}{Lemma \ref{l1}}
Indeed, the proof of Proposition \ref{Pack} in \cite{CK}
employs only the property $\frac{\delta n}{\log n} $ goes to $\infty$ as $n \to \infty$
even when $\delta$ depends on $n$.
Hence, we can apply Proposition \ref{Pack}  to the case with $\delta=n^{-\frac{3}{4}}$.
Now, we apply Proposition \ref{Pack} to the case when 
${\cal Y}={\cal X}$,
$\hat{\bV}(x|x')=\delta_{x,x'}$ and $\delta=n^{-\frac{3}{4}}$.
Since $x^n_i \in T_P$, we have ${\cal V}(P, {\cal Y})= V(x_i^n, {\cal X}) $.
We also have $T_{\hat{\bV}}(x_j^n)=\{x_j^n\}$ for any $j$.
The relation $I(P, \hat{\bV}) = H(P)$ implies that $-n |  I(P, \hat{\bV}) -R|^+= -(nH(P)-R) $.
So, \eqref{Eq5-24} is the same as \eqref{20}.
Hence, we obtain Lemma \ref{l1}.
\end{proofof}

\end{document}